%% file: report.tex
\documentclass[%
reprint,
amsmath,amssymb,
aps,
amsthm,
pra,
]{revtex4-1}

\usepackage{ifthen}
\newboolean{arxiv}
\setboolean{arxiv}{true}
\ifthenelse{\boolean{arxiv}}{\pdfoutput=1}{}

\newboolean{letter}
\setboolean{letter}{false}
\newcommand{\Letter}[2]{\ifthenelse{\boolean{letter}}{#1}{#2}}

\usepackage{theorem} 
\usepackage{ascmac}
\usepackage{bbm}
\usepackage{bm} 
\usepackage{braket}
\usepackage{dcolumn} 
\usepackage{enumerate}
\usepackage{enumitem}
\usepackage{framed}
\usepackage{graphicx} 
\ifthenelse{\boolean{arxiv}}{\usepackage{hyperref}}{\usepackage[dvipdfmx]{hyperref}}
\usepackage{longtable}
\usepackage{mathtools}
\usepackage{multirow}
\ifthenelse{\boolean{arxiv}}{}{\usepackage{pxjahyper}}
\usepackage{txfonts} 
\usepackage{xparse}

\usepackage{color}
\usepackage[most]{tcolorbox}

\ifthenelse{\boolean{arxiv}}{}{\mathtoolsset{showonlyrefs}}

\newcommand{\hypercolor}{blue}
\hypersetup{
  colorlinks,
  citecolor=\hypercolor,
  linkcolor=\hypercolor,
  urlcolor=\hypercolor,
  bookmarksopen,
  bookmarksopenlevel=4,
  bookmarksnumbered
}

\newboolean{figure}
\setboolean{figure}{true}
\newcommand{\InsertPDF}[2]{\iffigure\includegraphics[scale=#1]{figures/#2}\fi}

\input{settings_quant.tex}

\newcommand{\titlename}{Identifying quantum change points for Hamiltonians}
\newcommand{\supplementaltitle}{{\large \bf Supplemental Material for ``\titlename''} \\}
\newcommand{\supplementalauthor}[1]{\vspace*{0.5cm}#1 \\}
\newcommand{\supplementalaffiliation}[1]{\textit{#1}}
\newcommand{\citesupF}[1]{\Letter{Sec.~\ref{#1} of the Supplemental Material (SM)}{Appendix~\ref{#1}}}
\newcommand{\citesup}[1]{\Letter{Sec.~\ref{#1} of the SM}{Appendix~\ref{#1}}}

\newcommand{\Pna}{P_{\mathrm{na}}}
\newcommand{\ChoiE}{\mathfrak{E}}
\newcommand{\ChoiD}{\mathfrak{D}}
\newcommand{\ChoiU}{\mathfrak{U}}
\renewcommand{\ident}{\mathbbm{1}}
\newcommand{\I}{I}
\newcommand{\V}{V}
\newcommand{\W}{W}
\newcommand{\Pos}{\mathsf{Pos}}
\newcommand{\Chn}{\mathsf{Chn}}
\newcommand{\Den}{\mathsf{Den}}
\newcommand{\Tester}{\mathsf{Tester}}
\newcommand{\Comb}{\mathsf{Comb}}
\newcommand{\Realp}{\Real_{\ge 0}}
\newcommand{\ot}{\otimes}
\newcommand{\tmC}{\tilde{\mC}}
\newcommand{\tk}{\tilde{k}}

\newcommand{\tU}{\tilde{U}}
\newcommand{\Cone}{\mathsf{Cone}}
\newcommand{\w}{\omega}
\newcommand{\pH}{p^\mathrm{H}}

\newcommand{\Ad}{\mathrm{Ad}}
\newcommand{\bin}[1]{{\llbracket #1 \rrbracket}}
\let\ast\relax
\DeclareMathOperator{\ast}{\circledast}
\setlist[enumerate]{label=\arabic*), leftmargin=3em, itemsep=0pt, parsep=0pt, labelwidth=5em}

\definecolor{memo}{RGB}{128,0,255}
\definecolor{gray}{RGB}{128,128,128}

\newcommand{\Discard}[1]{}

\newcommand{\EN}[1]{}

\begin{document}

\preprint{APS/123-QED}

\title{\titlename}

\affiliation{%
 Quantum Information Science Research Center, Quantum ICT Research Institute, Tamagawa University,
 Machida, Tokyo 194-8610, Japan
}%

\author{Kenji Nakahira}
\affiliation{%
 Quantum Information Science Research Center, Quantum ICT Research Institute, Tamagawa University,
 Machida, Tokyo 194-8610, Japan
}%

\date{\today}

\begin{abstract}
 The identification of environmental changes is crucial in many fields.
 The present research is aimed at investigating the optimal performance for detecting change points
 in a quantum system when its Hamiltonian suddenly changes at a specific time.
 Assume that the Hamiltonians before and after the change are known
 and that the prior probability of each prospective change point is identical.
 These Hamiltonians can be time-dependent.
 The problem considered in this study is an extension of the problem of
 discriminating multiple quantum processes that consist of sequences of quantum channels.
 Although it is often extremely difficult to find an analytical solution to such a problem,
 we demonstrate that the maximum success probability for the Hamiltonian change point problem
 can be determined analytically and has a simple form. 
\end{abstract}

\pacs{03.67.Hk}
\maketitle



The detection of transition points is a critical issue that emerges in many fields.
In various physical situations, the relevant environment may undergo a sudden change,
and determining the precise moment of this transformation can yield valuable insights.
This challenge, known as the change point problem, has attracted considerable
statistical research interest
\cite{Page-1954,Page-1955,Bro-Dar-1993,Car-Mul-Sie-1994}.

In this study, we address the difficulty of identifying the precise instant
when the time evolution of the quantum system experiences sudden changes.
Figure~\ref{fig:change_point} shows the schematic of this problem.
Suppose a quantum system with Hamiltonian $H_0(t)$ undergoes a sudden change to the Hamiltonian $H_1(t)$
at time $t^\star$, such as when a magnetic field is suddenly applied at time $t^\star$,
or when the frequency of the applied magnetic field undergoes a sudden change.
These Hamiltonians may be time-dependent.
The exact values of $H_0(t)$ and $H_1(t)$ are known but the point of the transition occurrence,
i.e., $t^\star$, is unknown; therefore, it is crucial to determine $t^\star$ as accurately as possible.
Let us assume that there exist a finite number of transition points and that each candidate has
an equal chance of becoming a change point.
In this situation, we may input a particle with a known initial state into the system
and perform a quantum measurement on the output from the system.
We can perform any operation allowed by quantum mechanics,
such as modifying the Hamiltonian of the system by adding an external magnetic field
or replacing a particle in the system with another one at an arbitrary time.

The upper and lower bounds for the maximum success probability in the change point problem
for quantum states were established in Ref.~\cite{Sen-Bag-Cal-Chi-2016}.
In this paper, we discuss the challenge of detecting transitions between Hamiltonians,
requiring the identification of quantum processes as a sequence of unitary channels
that represent discrete time evolution.
It is significantly more difficult to distinguish quantum processes than quantum states
because it necessitates optimizing not only output measurements but also input states
and channels used during the process.
In the context of Ref.~\cite{Sen-Bag-Cal-Chi-2016}, the problem is limited to the case of pure states,
in which assuming them to be qubit states does not lose generality.
Nevertheless, the quantum system under consideration can have dimensions beyond two.
To obtain the optimal performance, it is necessary to consider various types of
distinguishing strategies, including those that combine entangled input states with ancillary systems
and adaptive strategies.
\begin{figure}[b]
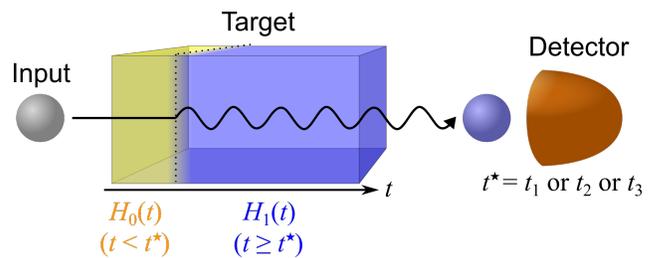

 \centering
 \InsertPDF{1.0}{change_point.png}
 \caption{Problem of discriminating change points of Hamiltonians.
 The Hamiltonian of a quantum system suddenly changes from $H_0(t)$ to $H_1(t)$ at some time $t = t^\star$.
 Assuming that the candidate change points are given, we wish to identify $t^\star$
 as accurately as possible by optimizing the state input to the system, the measurement for the output,
 and so on.}
 \label{fig:change_point}
\end{figure}

The problem of distinguishing between quantum (memoryless) channels has received significant attention
\cite{Sac-2005,Sac-2005-EB,Dua-Fen-Yin-2007,Li-Qiu-2008,Har-Has-Leu-Wat-2010,Mat-Pia-Wat-2010,
Dua-Guo-Li-Li-2016,Pir-Lup-2017,Li-Zhe-Sit-Qiu-2017,Pir-Lau-Lup-Per-2019,Kat-Wil-2020}.
Analytical solutions have already been identified for distinguishing between two simple quantum processes
or processes with significant symmetry, such as those covariant with respect to unitary operators.
However, our circumstance entails a time-dependent Hamiltonian, which cannot be categorized
as a quantum channel discrimination issue.
Moreover, to achieve our objectives, it is necessary to differentiate multiple processes
that do not exhibit high symmetry.
Thus, obtaining an analytical solution may be difficult.
In recent years, researchers have investigated the problem of distinguishing between
distinct forms of time-dependent, generalized quantum processes, also known as quantum memory channels
or quantum strategies
\cite{Gut-Wat-2007,Chi-Dar-Per-2008-memory,Gut-2012,Chi-2012,Nak-Kat-2021,Nak-Kat-2021-general,Nak-2021-restricted}.
The formulation of the optimal performance as a semidefinite programming problem is known \cite{Chi-2012},
which often proves helpful in obtaining numerical solutions.
However, as the dimension of the system and the number of candidates increase,
the computational complexity increases exponentially,
limiting the feasibility of obtaining a solution to small-scale problems.
Also, the conclusions of Ref.~\cite{Chi-2012} cannot be explicitly applied
when the candidate Hamiltonians vary continuously.

Surprisingly, in the aforementioned Hamiltonian change point problem,
the optimal performance can be obtained analytically and expressed in a simple form
for an arbitrary number of candidate change points.
We should emphasize that even in the change point problem for quantum states, which intuitively
seems to be easier to solve, an analytical solution is only found in the limiting case
where the number of candidate change points is infinite \cite{Sen-Bag-Cal-Chi-2016}.
Starting with a simplified case of determining the sudden change from a unitary channel to the next,
we examined the transition between channels.
We derive an optimal analytical solution and demonstrate that adaptive strategies and
ancillary systems are not required for optimal discrimination.
This is in contrast to the discrimination of more than two channels, which generally requires
the use of adaptive strategies in conjunction with ancillary systems
\cite{Dua-Fen-Yin-2007,Bav-Mur-Qui-2021}.
Then, this result is applied to analytically obtain an optimal solution
to the Hamiltonian change point problem.

\emph{Identification of change points for unitary channels} ---
We first formalize the task of recognizing transition points within unitary channels.
We assume a unitary channel, where the first $n$ uses correspond to $\mU_0$,
and the remaining uses correspond to $\mU_1$.
The channels $\mU_0$ and $\mU_1$ are known and $n$ can be any integer between $0$ to $N$.
The objective is to precisely determine the value of $n$.
Let $\mU_{n < k}$ be $\mU_1$ if $n < k$, and $\mU_0$ otherwise.
Let $\mE_n$ be the process consisting of a sequence of $N$ channels $(\mU_{n < 1},\dots,\mU_{n < N})$.
This problem can be formulated as differentiating between the $N+1$ processes $\mE_0,\dots,\mE_N$.
For example, in the case of $N = 2$, there is a need to distinguish between the three possible sequences:
$\mE_0 = (\mU_1, \mU_1)$, $\mE_1 = (\mU_0, \mU_1)$, and $\mE_2 = (\mU_0, \mU_0)$.
Let $\V_k$ and $\W_k$, with equal dimensions, denote the input and output systems, respectively,
for the channel $\mU_{n < k}$.

The most general discrimination strategy, presented in Fig.~\ref{fig:process_discrimination}(a),
involves ancillary systems $\V'_1,\dots,\V'_N$.
We begin by preparing a bipartite system with initial conditions of $\V_1 \ot \V'_1$.
The first segment $\V_1$ is sent through the channel $\mU_{n<1}$, followed by a channel $\sigma_2$.
Subsequently, $\V_2$ is sent through the channel $\mU_{n<2}$, followed by a channel $\sigma_3$,
until $N$ steps have been completed.
The system $\W_N \ot \V'_N$ is then subjected to a quantum measurement,
$\Pi \coloneqq \{ \Pi_m \}_{m=0}^N$.
A collection $(\rho, \sigma_2, \dots, \sigma_N, \Pi)$ can be used to define
any quantum discrimination strategy allowed by quantum mechanics, including
an entanglement-assisted and/or adaptive one.
This collection of objects is known as a quantum tester \cite{Chi-Dar-Per-2008-memory}.
We want to find a discrimination strategy that maximizes the success probability.
The problem of obtaining the maximum success probability, denoted by $P^{(N)}$,
is an optimization problem over quantum testers,
which is formulated as a semidefinite programming problem \cite{Chi-2012}.
We assume $P^{(N)} < 1$, which means that $\mU_0$ and $\mU_1$ are not perfectly distinguishable
with a single evaluation.
\begin{figure}[bt]
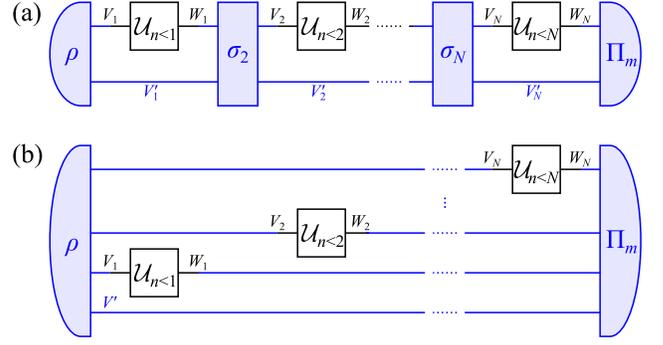

 \centering
 \InsertPDF{1.0}{process_discrimination.pdf}
 \caption{(a) The most general protocol of change point discrimination for unitary channels.
 Each process $\mE_n$ consists of a sequence of $N$ channels $(\mU_{n<1},\dots,\mU_{n<N})$.
 Any discrimination strategy is expressed as a collection of a state $\rho$,
 channels $\sigma_2,\dots,\sigma_N$, and a measurement $\{ \Pi_m \}_{m=0}^N$.
 (b) The most general nonadaptive protocol, which consists of a state $\rho$ and a measurement $\{ \Pi_m \}_{m=0}^N$.}
 \label{fig:process_discrimination}
\end{figure}

Figure~\ref{fig:process_discrimination}(b) shows the most general nonadaptive protocol,
which can be regarded as a special case of the protocol shown in Fig.~\ref{fig:process_discrimination}(a).
An initial state, $\rho$, is prepared for the multipartite system $\V_N \ot \cdots \ot \V_1 \ot \V'$
of the protocol, where $\V'$ is an ancillary system.
Subsystems $\V_1,\dots,\V_N$ are then exposed to their respective channels $\mU_{n<1},\dots,\mU_{n<N}$
and the system $\W_N \ot \cdots \ot \W_1 \ot \V'$ is subjected to a quantum measurement,
$\Pi \coloneqq \{ \Pi_m \}_{m=0}^N$.
Any nonadaptive discrimination strategy can be described using a collection $(\rho, \Pi)$.
The determination of the optimal performance by using only a nonadaptive strategy is simple;
however, the resulting performance could be inferior to those obtained using adaptive strategies.

Considering the optimal performance of nonadaptive strategies, let $\Lambda_n$ represent
the unitary channel composed of $N$ unitary channels $\mU_{n<N}, \dots, \mU_{n<1}$
connected in parallel, i.e.,
\begin{alignat}{1}
 \Lambda_n &\coloneqq \mU_{n<N} \ot \cdots \ot \mU_{n<1} = \mU_1^{\ot(N-n)} \ot \mU_0^{\ot n}.
\end{alignat}
This problem can be expressed as the following optimization problem:
\begin{alignat}{1}
 \begin{array}{ll}
  \mbox{maximize} & \displaystyle \frac{1}{N+1}
   \sum_{n=0}^N \Tr[\Pi_n \cdot (\Lambda_n \ot \ident_{V'})(\rho)], \\
 \end{array}
 \tag{$\mathrm{P_{na}}$} \label{prob:Pna}
\end{alignat}
where the maximization is taken over all possible input states $\rho$ of the system
$\V_N \ot \cdots \ot \V_1 \ot \V'$
and over all possible measurements $\Pi$ of the system $\W_N \ot \cdots \ot \W_1 \ot \V'$.
We denote the optimal value of this problem, i.e., the maximum success probability, by $\Pna^{(N)}$.
$\Pna^{(N)} \le P^{(N)}$ clearly holds.
For each $b \in \{0,1\}$, the channel $\mU_b$ is associated with a unitary matrix, $U_b$,
such that $\mU_b(\rho) = U_b \rho U_b^\dagger$.

We first consider the simplest case $N = 1$; then,
the problem is reduced to distinguishing two unitary channels $\mU_0$ and $\mU_1$ in a single trial,
and thus $P^{(1)} = \Pna^{(1)}$ holds.
$\Pna^{(1)}$ was obtained in Refs.~\cite{Chi-Pre-Ren-1999,Dar-Pre-Par-2001};
we here briefly review their results.
The maximum success probability for distinguishing two output states
$\mU_0 \ket{\psi}$ and $\mU_1 \ket{\psi}$ from a pure input state $\ket{\psi}$ is
given by
\begin{alignat}{1}
 \frac{1}{2} \left( 1 + \sqrt{1
 - \left| \braket{\psi|U_0^\dagger U_1|\psi} \right|^2} \right).
 \label{eq:P1}
\end{alignat}
The optimal input state $\ket{\psi}$ minimizes the absolute value of
the inner product of $U_0\ket{\psi}$ and $U_1\ket{\psi}$.
Let $\Gamma$ be the polygon in the complex plane whose vertices are the eigenvalues of
the unitary matrix $U_0^\dagger U_1$; then,
the distance between the polygon $\Gamma$ and the origin is equal to the minimum value of
$|\braket{\psi|U_0^\dagger U_1|\psi}|$
\footnote{Let $\ket{\lambda_k}$ be the normalized eigenvector corresponding to
the eigenvalue $\lambda_k$ of $U_0^\dagger U_1$; then,
we have $\braket{\psi|U_0^\dagger U_1|\psi} = \sum_k |\braket{\lambda_k|\psi}|^2 \lambda_k$.
Since $\Gamma$ can be represented by $\{ \sum_k w_k \lambda_k : w_k \ge 0, ~\sum_k w_k = 1 \}$,
the minimum value of $|\braket{\psi|U_0^\dagger U_1|\psi}|$ is equal to the distance between
$\Gamma$ and the origin.
This fact was noticed in footnote~[22] of Ref.~\cite{Dar-Pre-Par-2001}.}.
If $\lambda_0$ and $\lambda_1$ are the eigenvalues representing the points at both ends of
the polygon $\Gamma$ closest to the origin,
and $\ket{\lambda_0}$ and $\ket{\lambda_1}$ are the corresponding normalized
eigenvectors, then $|\braket{\psi|U_0^\dagger U_1|\psi}|$ attains its minimum value
$|\lambda_0 + \lambda_1| / 2$ at
$\ket{\psi} = \ket{\pm} \coloneqq (\ket{\lambda_0} \pm \ket{\lambda_1}) / \sqrt{2}$.
Thus, from Eq.~\eqref{eq:P1}, we obtain
\begin{alignat}{1}
 P^{(1)} &= \frac{\gamma+1}{2}, \quad
 \gamma \coloneqq \frac{1}{2} \left| \lambda_1 - \lambda_0 \right|.
\end{alignat}

We next consider the case $N = 2$, where the task is to discriminate between
$\mE_0 = (\mU_1, \mU_1)$, $\mE_1 = (\mU_0, \mU_1)$, and $\mE_2 = (\mU_0, \mU_0)$.
Let us concentrate on discrimination strategies without any ancillary system.
When $\V_2 \ot \V_1$ is prepared in an initial state $\ket{\psi}$,
the output states of the three processes are
$\ket{\psi'_0} \coloneqq U_1 \ot U_1 \ket{\psi}$,
$\ket{\psi'_1} \coloneqq U_1 \ot U_0 \ket{\psi}$,
and $\ket{\psi'_2} \coloneqq U_0 \ot U_0 \ket{\psi}$.
To minimize $|\braket{\psi'_1|\psi'_0}|$ and $|\braket{\psi'_2|\psi'_1}|$, $\ket{\psi} = \ket{+}\ket{+}$ or $\ket{\psi} = \ket{-}\ket{-}$ could be selected,
in which it seems likely that a high success probability would be obtained.
Furthermore, we find that when $\ket{\psi}$ is an entangled state represented in the following form,
$|\braket{\psi'_1|\psi'_0}|$ and $|\braket{\psi'_2|\psi'_1}|$ attain minimum values:
\begin{alignat}{1}
 \ket{\psi} &= a_1 \ket{+}\ket{+} + a_2 \ket{-}\ket{-},
 \quad |a_1|^2 + |a_2|^2 = 1.
 \label{eq:psi2}
\end{alignat}
A higher success probability may be obtained using such an entangled state.
Also, while the problem with $N = 1$ can be just reduced to minimizing the absolute value of
the inner product of $U_0 \ket{\psi}$ and $U_1 \ket{\psi}$,
the problem with $N = 2$ becomes more complicated.
In fact, there is no guarantee that the smaller the values
$|\braket{\psi'_1|\psi'_0}|$, $|\braket{\psi'_2|\psi'_1}|$,
and $|\braket{\psi'_2|\psi'_0}|$, the higher the success probability
[see \citesupF{sec:iprod}].
Moreover, the use of an ancillary system may increase the success probability.
It would also not be surprising if $P^{(2)}$ were strictly larger than $\Pna^{(2)}$.

Considering the above discussion, obtaining an analytical expression of $P^{(N)}$ for $N \ge 2$
is challenging.
However, we discovered that a nonadaptive strategy can achieve optimal discrimination for any $N$.
The following theorem provides a simple expression for the exact value of $P^{(N)}$ 
as a function of $N$ and $\gamma$.
\begin{thmMain} \label{thm:main_eq}
 In the problem of change point discrimination for unitary channels, we have
 \begin{alignat}{1}
  P^{(N)} &= \Pna^{(N)} = \frac{N\gamma + 1}{N + 1}.
 \end{alignat}
\end{thmMain}
\begin{proof}
 We present a summary of the proof; for more details, please refer to \citesup{sec:thm}.
 Consider a nonadaptive discrimination strategy in which pure state $\ket{\psi}$
 of the system $\V_N \ot \cdots \ot \V_1$ is input into the channel $\Lambda_n$.
 Assume that the measurement in an orthonormal system $\{ \ket{\pi_m} \}_{m=0}^N$ identifies
 the pure output state from the channel $\Lambda_n$, which is represented by
 $[U_1^{\ot(N-n)} \ot U_0^{\ot n}] \ket{\psi}$; then, the conditional probability,
 denoted by $p_{m|n}$, that the measurement result is $m$, given that the change point is $n$ is represented as 
 \begin{alignat}{1}
  p_{m|n} &= \left| \braket{\pi_m| [U_1^{\ot(N-n)} \ot U_0^{\ot n}] |\psi} \right|^2.
 \end{alignat}
 Let us choose
 \begin{alignat}{1}
  \ket{\psi} &= \sum_{s_1 \in \{-1,1\}} \dots \sum_{s_N \in \{-1,1\}}
  a_{s_1,\dots,s_N} \ket{\phi_{s_N}} \cdots \ket{\phi_{s_1}}
  \label{eq:psi_main}
 \end{alignat}
 with $\ket{\phi_1} \coloneqq \ket{+}$ and $\ket{\phi_{-1}} \coloneqq \ket{-}$,
 where $a_{s_1,\dots,s_N}$ is
 $\prod_{k=1}^{N-1} [(s_k s_{k+1} \sqrt{\gamma} + 1) / \sqrt{2(\gamma+1)}]$
 if the sequence $s_1,\dots,s_N$ has an even number of elements equal to $-1$, and $0$
 otherwise.
 Note that in the case of $N = 2$, such $\ket{\psi}$ is in the form of Eq.~\eqref{eq:psi2}.
 It can be seen that there exists an orthonormal system $\{ \ket{\pi_m} \}_{m=0}^N$
 satisfying
 \begin{alignat}{1}
  p_{m|n} &=
  \begin{dcases}
   \sum_{k=-\infty}^0 \mL(k - n; \zeta), & m = 0, \\
   \mL(m - n; \zeta), & 0 < m < N, \\
   \sum_{k=N}^\infty \mL(k - n; \zeta), & m = N, \\
  \end{dcases}
  \label{eq:pnm2}
 \end{alignat}
 where $\zeta \coloneqq (1-\gamma) / (1+\gamma)$ and
 $\mL(n;\zeta)$ is the probability mass function of the discrete Laplace distribution,
 given by
 \begin{alignat}{1}
  \mL(n;\zeta) &\coloneqq \frac{1-\zeta}{1+\zeta} \zeta^{|n|}, \quad 0 < \zeta < 1.
 \end{alignat}
 This nonadaptive strategy provides the success probability of
 \begin{alignat}{1}
  \frac{1}{N+1} \sum_{n=0}^N p_{n|n} &= \frac{N\gamma + 1}{N + 1} \eqqcolon q,
 \end{alignat}
 which is clearly not greater than $\Pna^{(N)}$.
 In addition, we can demonstrate that the optimal value, denoted by $D$, of the Lagrange dual
 of the change point problem is upper bounded by $q$.
 As the weak duality inequality $P^{(N)} \le D$ holds, we have
 \begin{alignat}{1}
  q &\le \Pna^{(N)} \le P^{(N)} \le D \le q,
 \end{alignat}
 and thus all these inequalities are equalities.
\end{proof}

In practice, it may be challenging to implement the entangled states described by Eq.~\eqref{eq:psi_main}.
Alternatively, consider using a separable state, $\ket{+}^{\ot N}$, as an input.
In this situation, the problem is to identify the state
$\{ (U_1 \ket{+})^{\ot(N-n)} (U_0 \ket{+})^{\ot n} \}_n$,
and thus reduces to a quantum change point problem for pure states, which has been studied in
Ref.~\cite{Sen-Bag-Cal-Chi-2016}.
However, the use of a separable state results in performance degradation,
as shown in Fig.~\ref{fig:result_compare}.
\begin{figure}[bt]
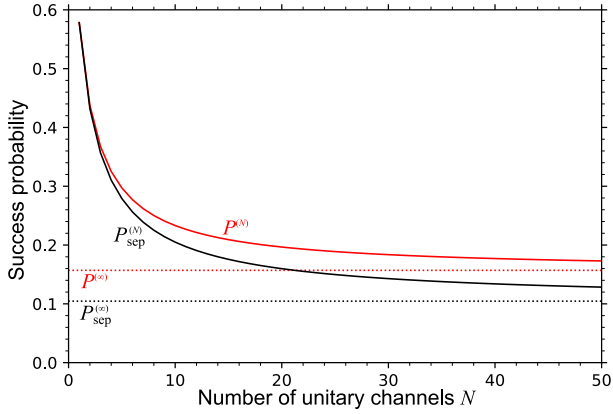

 \centering
 \InsertPDF{0.6}{result-compare.pdf}
 \caption{Probability of successful identification of change points
 for $\lambda_0 = 1$ and $\lambda_1 = \exp(\i \pi / 10)$.
 $P^{(N)}_\mathrm{sep}$ is the maximum success probability when the separable state $\ket{+}^{\ot N}$
 is used as an input, which is obtained by numerically solving a semidefinite programming problem.
 $P^{(\infty)} = \gamma$ and $P^{(\infty)}_\mathrm{sep}$ are limits as $N \to \infty$.
 The analytical solution of $P^{(\infty)}_\mathrm{sep}$ is given in Ref.~\cite{Sen-Bag-Cal-Chi-2016}.}
 \label{fig:result_compare}
\end{figure}

\emph{Identification of change points for Hamiltonians} ---
The above-mentioned discussion can be expanded to address the problem of identifying
Hamiltonian change points.
Consider a situation where a Hamiltonian $H_0(t)$ acting on a quantum system changes to $H_1(t)$
at a particular time, $t^\star$.
Assume that the change point $t^\star$ is known to be one of the possible candidates $t_0,\dots,t_N$
with equal prior probabilities.
We arbitrarily choose a natural number $R$ and time instants
$\tau_0 < \tau_1 < \dots < \tau_{NR}$ such that
$\tau_{kR} = t_k$ holds for each $k \in \{0,\dots,N\}$.
For each $b \in \{0,1\}$, $1 \le k \le N$, and $1 \le r \le R$,
let $\mU^{(k,r)}_b$ be the unitary channel representing
the time evolution with the Hamiltonian $H_b(t)$ between the time interval
$\tau_{(k-1)R + r - 1} \le t \le \tau_{(k-1)R + r}$, i.e.,
\begin{alignat}{1}
 \mU^{(k,r)}_b(\rho) &\coloneqq U^{(k,r)}_b \rho U^{(k,r)}_b{}^\dagger, \nonumber \\
 \quad U^{(k,r)}_b &\coloneqq \mT \left[ \exp
 \left[ - \i \int_{\tau_{(k-1)R + r - 1}}^{\tau_{(k-1)R + r}} H_b(t) dt \right] \right],
\end{alignat}
where $\mT$ is the time-ordered operator.
Also, let $\mE_n$ be the sequence expressed by
\begin{alignat}{1}
 \mE_n &\coloneqq [\mU^{(1,1)}_{n<1},\mU^{(1,2)}_{n<1},\dots,\mU^{(1,R)}_{n<1}, \nonumber \\
 &\qquad \mU^{(2,1)}_{n<2},\mU^{(2,2)}_{n<2},\dots,\mU^{(2,R)}_{n<2}, \dots, \nonumber \\
 &\qquad \mU^{(N,1)}_{n<N},\mU^{(N,2)}_{n<N},\dots,\mU^{(N,R)}_{n<N}].
\end{alignat}
The problem is then reduced to distinguishing quantum processes $\mE_0,\dots,\mE_N$
similar to the problem of unitary channels.
However, the main difference in this case is that the Hamiltonians can be time-dependent
and we can use arbitrarily short time intervals.
As a result, this problem is challenging to solve analytically.

Let $\mu_\mathrm{max}(t)$ and $\mu_\mathrm{min}(t)$, respectively, be
the maximum and minimum eigenvalues of $H_1(t) - H_0(t)$.
Also, let
\begin{alignat}{1}
 \gamma_k &\coloneqq \Delta \left[
 \int_{t_{k-1}}^{t_k} [\mu_\mathrm{max}(t) - \mu_\mathrm{min}(t)] dt \right],
 \quad 1 \le k \le N,
\end{alignat}
where $\Delta(\theta) \coloneqq \sin[\min(\theta,\pi)/2]$.

For $N = 1$, the problem is to identify two processes $\mE_0$ and $\mE_1$.
In this problem, the maximum success probability is known as $(\gamma_1 + 1) / 2$
\cite{Chi-Pre-Ren-1999,Aha-Mas-Pop-2002}, which is obtained as limit $R \to \infty$.
In addition, some experiments have been conducted using this result \cite{Sch-Gef-Lou-Ost-2021}.
We find that by extending the proof of Theorem~\ref{thm:main_eq},
an analytical expression of the ultimate performance for each $N$ is obtained,
as stated in the following theorem (the proof is given in \citesup{subsec:thm_hamiltonian}).
\begin{thmMain} \label{thm:main_hamiltonian}
 The maximum success probability in the change point problem for two Hamiltonians
 $H_0(t)$ and $H_1(t)$ with any integer $N \ge 1$ is given by
 \begin{alignat}{1}
  \frac{1}{N + 1} \left( \sum_{k=1}^N \gamma_k + 1 \right).
  \label{eq:PN_hamiltonian}
 \end{alignat}
\end{thmMain}

In the limit of large $N$, it follows from Eq.~\eqref{eq:PN_hamiltonian} that
the maximum success probability tends to the average of $\gamma_1,\gamma_2,\dots,\gamma_N$
(i.e., $\sum_{k=1}^N \gamma_k / N$).

\emph{Conclusions} ---
In this paper, we investigated the difficulty of identifying a precise moment
when the Hamiltonian suddenly changes, and we presented an analytical expression of
the maximum success probability.
We first discussed the quantum change point problem for unitary channels, as a simpler problem.
The objective of this task is to accurately identify the exact moment
when a unitary channel changes to another.
We demonstrated that the maximum success probability can be expressed in a simple analytical form
by using only the number of possible change points and a parameter reflecting
the ease of recognizing the channels before and after the change,
assuming identical prior probabilities.
The proposed method was then applied to derive the optimal performance for the problem of
discriminating change points for Hamiltonians.

This work lays the foundation for future research on related topics,
including the estimation of a continuous-valued change point and the detection of multiple change points.
In addition, it can facilitate research on the change point problem for channels in open systems
(i.e., non-unitary channels) and the optimization with other criteria
such as unambiguous or Neyman-Pearson.
We anticipate that our results will provide a solid starting point for addressing these challenges.

We thank for O.~Hirota, M.~Sohma, T.~S.~Usuda, and K.~Kato for insightful discussions.
This work was supported by the Air Force Office of Scientific Research under
award number FA2386-22-1-4056.


\clearpage\onecolumngrid

\Letter{
\pagestyle{empty}
\setcounter{equation}{0}
\renewcommand{\theequation}{S\arabic{equation}}
\setcounter{section}{0}
\renewcommand{\citenumfont}[1]{S#1}
\renewcommand\bibnumfmt[1]{[S#1]}
\setcounter{thm}{0}
\renewcommand{\thethm}{S\arabic{thm}}
\setcounter{figure}{0}
\renewcommand{\thefigure}{S\arabic{figure}}

\begin{center}
 \supplementaltitle%
 \supplementalauthor{Kenji Nakahira}%
 \supplementalaffiliation{%
 Quantum Information Science Research Center, Quantum ICT Research Institute, Tamagawa University,
 Machida, Tokyo 194-8610, Japan
 }%
\end{center}
}{
\appendix
}

\section{Notation}

Let us set up some notation.
Let $\Real$, $\Realp$, and $\Complex$ be, respectively, the sets of all real numbers,
all nonnegative real numbers, and all complex numbers.
For a unitary matrix $U$, let $\Ad_U$ be the unitary channel determined by
$\Ad_U(\rho) = U \rho U^\dagger$.
We call a completely positive map a single-step process.
$\Pos(\V,\W)$ and $\Chn(\V,\W)$, respectively, denote the sets of all single-step processes
and quantum channels (i.e., completely positive trace-preserving maps) from a system $\V$ to a system $\W$.
Also, let $\Pos_\V$ and $\Den_\V$ be, respectively, the sets
of all positive semidefinite matrices and states (i.e., density matrices) of a system $\V$.
Let $\Cone(\mX)$ be the convex cone spanned by a set $\mX$.
$\V \cong \V'$ denotes that the dimensions of the systems $\V$ and $\V'$ are the same.
Let $\I_\V$ be the identity matrix on $\V$ and we write $\ident_\V \coloneqq \Ad_{\I_V}$.
$X \ge Y$ denotes that $X - Y$ is positive semidefinite.

A single-step process $f \in \Pos(\V,\W)$ is depicted by
\begin{alignat}{1}
 \InsertPDF{1.0}{single_step_process.pdf} ~\raisebox{.1em}{,}
 \label{eq:single_step_process_pdf}
\end{alignat}
where systems are depicted by labeled wires.
Wires representing one-dimensional systems will often be omitted.
Single-step processes can be composed sequentially or in parallel.
Let us consider the concatenation of $T$ single-step processes,
which we call a $T$-step process,
$\{ \mE^{(t)} \in \Pos(\W'_{t-1} \ot \V_t, \W'_t \ot \W_t) \}_{t=1}^T$
(where $W'_0$ and $W'_T$ are one-dimensional) such as
\begin{alignat}{1}
 \InsertPDF{1.0}{comb.pdf} ~\raisebox{.1em}{.}
 \label{eq:comb_pdf}
\end{alignat}
We write this $T$-step process by $\mE^{(T)} \ast \cdots \ast \mE^{(1)}$,
where $\ast$ denotes the concatenation.
A $T$-step process $\mE^{(T)} \ast \cdots \ast \mE^{(1)}$ is called
a quantum comb \cite{Chi-Dar-Per-2008} if $\mE^{(1)},\dots,\mE^{(T)}$ are quantum channels.
In particular, this manuscript often focuses on a $T$-step process
generated by $T$ unitary channels, which is depicted by
\begin{alignat}{1}
 \InsertPDF{1.0}{comb_unitary.pdf} ~\raisebox{.1em}{,}
 \label{eq:comb_unitary_pdf}
\end{alignat}
where $\{ \mU^{(t)} \in \Chn(\V_t, \W_t) \}_{t=1}^T$ are unitary channels.
Given a $T$-step process written by Eq.~\eqref{eq:comb_unitary_pdf},
we consider a collection of $(T+1)$-step processes,
denoted by $\{ \mD_m \coloneqq \Pi_m \ast \mD^{(T)} \ast \cdots \ast \mD^{(1)} \}_{m=0}^M$,
where each $\mD_m$ is depicted by
\begin{alignat}{1}
 \InsertPDF{1.0}{comb2.pdf} ~\raisebox{1em}{.}
 \label{eq:comb2_pdf}
\end{alignat}
$\{ \mD_m \}_{m=0}^M$ is called a tester if $\mD^{(1)},\dots,\mD^{(T)}$ are quantum channels
[which implies that $\mD^{(1)}$ is a state] and $\{ \Pi_m \}_{m=0}^M$ is a quantum measurement.
We will depict tester elements in blue.

The Choi-Jamio{\l}kowski representations \cite{Cho-1975,Jam-1972,Chi-Dar-Per-2008} of
processes will be denoted by the same letter in the Fraktur font.
The Choi-Jamio{\l}kowski representation, $\ChoiE$, of the $T$-step process $\mE$
written by Eq.~\eqref{eq:comb_pdf} is the state of the system
$\W_T \ot \V_T \ot \cdots \ot \W_1 \ot \V_1$
that is depicted by
\begin{alignat}{1}
 \InsertPDF{1.0}{comb_choi.pdf} ~\raisebox{1em}{,}
 \label{eq:comb_choi}
\end{alignat}
where $\Psi_t \coloneqq \kket{\I_{\V_t}} \bbra{\I_{\V_t}}$
and $\kket{\I_{\V_t}} \coloneqq \sum_n \ket{n}\ket{n} \in \V_t \ot \V_t$.
If a $T$-step process $\mE$ can be written in the form of Eq.~\eqref{eq:comb_unitary_pdf},
then its Choi-Jamio{\l}kowski representation is $\ChoiE = \bigotimes_{t=0}^{T-1} \ChoiU^{(T-t)}$,
where $\ChoiU^{(t)}$ is the Choi-Jamio{\l}kowski representation of $\mU^{(t)}$,
i.e., $\ChoiU^{(t)} \coloneqq [\mU^{(t)} \ot \ident_{\V_t}] \Psi_t \in \Pos_{\W_t \ot \V_t}$.
Let $\Comb_{\W_T,\V_T;\dots;\W_1,\V_1}$ be the set of
all $\tau \in \Pos_{\W_T \ot \V_T \ot \cdots \ot \W_1 \ot \V_1}$
for which there exists a collection
$\{ \tau^{(t)} \in \Pos_{\W_t \ot \V_t \ot \cdots \ot \W_1 \ot \V_1} \}_{t=1}^T$
satisfying
\begin{alignat}{1}
 \tau^{(T)} &= \tau, \nonumber \\
 \Trp{\W_t} \tau^{(t)} &= \I_{\V_t} \ot \tau^{(t-1)}, \quad \forall t \in \{ 2,\dots,T \}, \nonumber \\
 \Trp{\W_1} \tau^{(1)} &= \I_{\V_1}.
 \label{eq:Comb}
\end{alignat}
A $T$-step process $\mE$ written in the form of Eq.~\eqref{eq:comb_pdf} is a quantum comb
if and only if its Choi-Jamio{\l}kowski representation $\ChoiE$ is in $\Comb_{\W_T,\V_T;\dots;\W_1,\V_1}$.
Let $\Tester^{(M)}_{\W_T,\V_T;\dots;\W_1,\V_1}$ be the set of
all collections of $M+1$ positive semidefinite matrices
$\{ \tau'_m \in \Pos_{\W_T \ot \V_T \ot \cdots \ot \W_1 \ot \V_1} \}_{m=0}^M$
satisfying $\sum_{m=0}^M \tau'_m \in \Comb_{\Complex,\W_T;\V_T,\dots,\W_1;\V_1,\Complex}$.
It follows that $\{ \tau'_m \}_{m=0}^M$ is in $\Tester^{(M)}_{\W_T,\V_T;\dots;\W_1,\V_1}$
if and only if there exists a collection
$\{ \tau^{(t)} \in \Pos_{\V_t \ot \W_{t-1} \ot \V_{t-1} \ot \cdots \ot \V_1} \}_{t=1}^T$
satisfying
\begin{alignat}{1}
 \I_{\W_T} \ot \tau^{(T)} &= \sum_{m=0}^M \tau'_m, \nonumber \\
 \Trp{\V_t} \tau^{(t)} &= \I_{\W_{t-1}} \ot \tau^{(t-1)}, \quad \forall t \in \{2,\dots,T\}, \nonumber \\
 \Tr \tau^{(1)} &= 1.
 \label{eq:Phi_sum}
\end{alignat}
Also, a collection of $(T+1)$-step processes $\{ \mD_m \}_{m=0}^M$
written in the form of Eq.~\eqref{eq:comb2_pdf} is a tester if and only if
$\{ \ChoiD_m \}_{m=0}^M \in \Tester^{(M)}_{\W_T,\V_T;\dots;\W_1,\V_1}$ holds.
For a quantum comb $\mE$ with $\ChoiE \in \Comb_{\W_T,\V_T;\dots;\W_1,\V_1}$ and
a tester $\{ \mD_m \}_{m=0}^M$
with $\{ \ChoiD_m \}_{m=0}^M \in \Tester^{(M)}_{\W_T,\V_T;\dots;\W_1,\V_1}$,
the probability that the tester $\{ \mD_m \}_{m=0}^M$ performed on $\mE$ gives the outcome $m$
is diagrammatically depicted by
\begin{alignat}{1}
 \InsertPDF{1.0}{comb_tester.pdf} ~\raisebox{1em}{,}
 \label{eq:comb_tester_pdf}
\end{alignat}
which is equal to $\mD_m \ast \mE = \Tr(\ChoiD_m^\T \ChoiE) \in \Realp$ (where $^\T$ is the transpose).
We can easily verify $\sum_{m=0}^M \Tr(\ChoiD_m^\T \ChoiE) = 1$.

\section{Problem Formulation} \label{sec:problem}

For each $k \in \{0,\dots,N\}$ and $b \in \{0,1\}$, let
\begin{alignat}{1}
 \mU^{(k)}_b &\coloneqq \mU^{(k,R)}_b \ast \mU^{(k,R-1)}_b \ast \cdots \ast \mU^{(k,1)}_b
 \label{eq:Ukb}
\end{alignat}
be an $R$-step process that consists of $R$ unitary channels
$\mU^{(k,R)}_b,\dots,\mU^{(k,1)}_b$.
The input and output systems of $\mU^{(k,r)}_b$ are, respectively,
denoted by $\V_{k,r}$ and $\W_{k,r}$.
$\mU^{(k)}_b$ is diagrammatically represented as
\begin{alignat}{1}
 \InsertPDF{1.0}{U.pdf}  ~\raisebox{.5em}{.}
 \label{eq:U_pdf}
\end{alignat}
Let us concentrate on the case in which
$\V_{k,r} \cong \W_{k,r} \cong \V_{1,1}$ holds for any $k$ and $r$.
We choose a unitary matrix $U^{(k,r)}_b$ such that $\mU^{(k,r)}_b = \Ad_{U^{(k,r)}_b}$.
The Choi-Jamio{\l}kowski representation, $\ChoiU^{(k)}_b$, of $\mU^{(k)}_b$ is written by
\begin{alignat}{1}
 \ChoiU^{(k)}_b &= \bigotimes_{r=0}^{R-1} \ChoiU^{(k,R-r)}_b
 = \ChoiU^{(k,R)}_b \ot \ChoiU^{(k,R-1)}_b \ot \cdots \ot \ChoiU^{(k,1)}_b,
\end{alignat}
where $\ChoiU^{(k,r)}_b \coloneqq [\mU^{(k,r)}_b \ot \ident_{\V_{k,r}}] (\kket{\I_{\V_{k,r}}}
\bbra{\I_{\V_{k,r}}})$ is the Choi-Jamio{\l}kowski representation of $\mU^{(k,r)}_b$.
Let us consider the following $NR$-step process:
\begin{alignat}{1}
 \mE^{(N)}_n &\coloneqq \mU^{(N)}_\bin{n < N} \ast \mU^{(N-1)}_\bin{n < N-1} \ast \cdots \ast \mU^{(1)}_\bin{n < 1},
 \quad n \in \{0,\dots,N\},
 \label{eq:E}
\end{alignat}
where $\bin{n < k}$ is $1$ if $n < k$, and $0$ otherwise.
$\mE^{(N)}_n$ is the process such that
$n$ processes $\mU^{(1)}_0,\mU^{(2)}_0,\dots,\mU^{(n)}_0$ are applied
and then $N-n$ processes $\mU^{(n+1)}_1,\dots,\mU^{(N)}_1$ are applied.
The problem of discriminating quantum processes $\mE^{(N)}_0,\dots,\mE^{(N)}_N$
can be seen as a generalization of the problem presented in the main paper.
We find that the change point problem for unitary channels described in the main paper can be viewed
as a special case of this problem with $R = 1$ and
$\mU^{(1)}_b = \cdots = \mU^{(N)}_b$ $~(\forall b \in \{0,1\})$.
Also, this problem with $R \to \infty$ corresponds to the change point problem for Hamiltonians.
The Choi-Jamio{\l}kowski representation, $\ChoiE^{(N)}_n$, of the process $\mE^{(N)}_n$ is
given by
\begin{alignat}{1}
 \ChoiE^{(N)}_n &= \bigotimes_{k=0}^{N-1} \ChoiU^{(N-k)}_\bin{n < N-k}
 = \ChoiU^{(N)}_\bin{n < N} \ot \ChoiU^{(N-1)}_\bin{n < N-1} \ot \cdots
 \ot \ChoiU^{(1)}_\bin{n < 1}.
 \label{eq:tmEN}
\end{alignat}
For example, in the case of $N = 2$, $\ChoiE^{(2)}_0 = \ChoiU^{(2)}_1 \ot \ChoiU^{(2)}_1$,
$\ChoiE^{(2)}_1 = \ChoiU^{(2)}_1 \ot \ChoiU^{(1)}_0$,
and $\ChoiE^{(2)}_2 = \ChoiU^{(2)}_0 \ot \ChoiU^{(1)}_0$ hold.

The problem of discriminating quantum combs $\mE_0,\dots,\mE_N$
with equal prior probabilities can be formulated as the following semidefinite programming problem:
\begin{alignat}{1}
 \begin{array}{ll}
  \mbox{maximize} & \displaystyle \frac{1}{N+1}
   \sum_{n=0}^N \Tr[\ChoiD_n^\T \ChoiE^{(N)}_n] \\
  \mbox{subject~to} & \displaystyle \left\{ \ChoiD_m \right\}_{m=0}^N \in
   \Tester^{(N)}_{\W_{N,R},\V_{N,R};\dots;\W_{N,1},\V_{N,1};\W_{N-1,R},\V_{N-1,R};\dots;\W_{1,1},\V_{1,1}}.
 \end{array}
 \tag{$\mathrm{P}$} \label{prob:P}
\end{alignat}
Let $P^{(N)}$ denote the optimal value of this problem.
The Lagrange dual problem is given by \cite{Chi-2012}
\begin{alignat}{1}
 \begin{array}{ll}
  \mbox{minimize} & \eta(X) \\
  \mbox{subject~to} & \displaystyle X \in \mC^{(N)},
   ~X \ge \frac{\ChoiE^{(N)}_m}{N+1} ~(\forall m \in \{0,\dots,N\}), \\
 \end{array}
 \tag{$\mathrm{D}$} \label{prob:D}
\end{alignat}
where
\begin{alignat}{1}
 \mC^{(N)} &\coloneqq
 \Cone(\Comb_{\W_{N,R},\V_{N,R};\dots;\W_{N,1},\V_{N,1};\W_{N-1,R},\V_{N-1,R};\dots;\W_{1,1},\V_{1,1}})
 \nonumber \\
 &= \left\{ t \chi : t \in \Realp,
 ~\chi \in \Comb_{\W_{N,R},\V_{N,R};\dots;\W_{N,1},\V_{N,1};\W_{N-1,R},\V_{N-1,R};\dots;\W_{1,1},\V_{1,1}} \right\}
\end{alignat}
and $\eta$ is the function defined as
\begin{alignat}{1}
 \eta(t \chi) = t, \quad \forall
 t \in \Realp, ~\chi \in \Comb_{\W_{N,R},\V_{N,R};\dots;\W_{N,1},\V_{N,1};\W_{N-1,R},\V_{N-1,R};\dots;\W_{1,1},\V_{1,1}}.
\end{alignat}
Note that, for each nonzero $X \in \mC^{(N)}$, $t \in \Realp$ and
$\chi \in \Comb_{\W_{N,R},\V_{N,R};\dots;\W_{N,1},\V_{N,1};\W_{N-1,R},\V_{N-1,R};\dots;\W_{1,1},\V_{1,1}}$
satisfying $X = t \chi$ are uniquely determined by $t = \eta(X)$ and $\chi = X / t$.
The optimal value of the dual problem coincides with the optimal value, $P^{(N)}$, of the primal problem
\cite{Chi-2012}.

\section{Maximum success probability} \label{sec:thm}

Let
\begin{alignat}{1}
 \tU^{(k,r)}_b &\coloneqq U^{(k,r)}_0{}^\dagger U^{(k,r)}_b =
 \begin{dcases}
  \I_{\V_{k,r}}, & b = 0, \\
  U^{(k,r)}_0{}^\dagger U^{(k,r)}_1, & b = 1.
 \end{dcases}
\end{alignat}
$\Gamma^{(k,r)}$ denotes the polygon in the complex plane whose vertices are
the eigenvalues of $\tU^{(k,r)}_1$.
Let $\lambda^{(k,r)}_0$ and $\lambda^{(k,r)}_1$ be the ends of a side of $\Gamma^{(k,r)}$
that is closest to the origin (see Fig.~\ref{fig:eigenvalues}).
Note that such a side may not be unique.
By swapping $\lambda^{(k,r)}_0$ and $\lambda^{(k,r)}_1$ if necessary,
we assume $0 \le \arg[\lambda^{(k,r)}_1 / \lambda^{(k,r)}_0] \le \pi$.
Let $\ket{\lambda^{(k,r)}_0}$ and $\ket{\lambda^{(k,r)}_1}$ be, respectively,
the normalized eigenvectors of $\tU^{(k,r)}_1$ corresponding to
the eigenvalues $\lambda^{(k,r)}_0$ and $\lambda^{(k,r)}_1$.
\begin{figure}[bt]
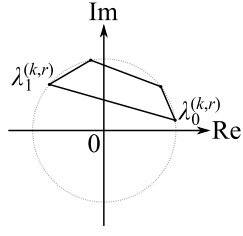

 \centering
 \InsertPDF{1.0}{eigenvalues.pdf}
 \caption{Example of two vertices, $\lambda^{(k,r)}_0$ and $\lambda^{(k,r)}_1$, of $\Gamma^{(k,r)}$,
 in which $\tU^{(k,r)}_1$ has four eigenvalues.
 The square with solid line, which is inscribed in the unit circle, represents $\Gamma^{(k,r)}$.}
 \label{fig:eigenvalues}
\end{figure}

For each $k \in \{1,\dots,N\}$, let $\gamma_k$ be $1$ if at least one of the polygons
$\Gamma^{(k,1)},\dots,\Gamma^{(k,R)}$
includes the origin or $\sum_{r=1}^R \arg[\lambda^{(k,r)}_1 / \lambda^{(k,r)}_0] \ge \pi$ holds,
and
\begin{alignat}{1}
 \gamma_k &\coloneqq
 \frac{1}{2} \left| \frac{\lambda^{(k)}_1}{\lambda^{(k)}_0} - 1 \right|,
 \quad \lambda^{(k)}_b \coloneqq \prod_{r=1}^R \lambda^{(k,r)}_b
\end{alignat}
otherwise.
Also, let $\gamma_0 \coloneqq 1$ and $\gamma_{N+1} \coloneqq 1$ for convenience.
In the case of $N = 1$, Problem~\eqref{prob:P} reduces to the problem of discriminating
$\mU^{(1)}_0$ and $\mU^{(1)}_1$ with equal prior probabilities in a single shot,
whose optimal value is known to be $(\gamma_1 + 1) / 2$
\cite{Chi-Pre-Ren-1999,Aha-Mas-Pop-2002,Dar-Pre-Par-2001}.
It is easily seen that for each $k$, the maximum success probability for discriminating
$\mU^{(k)}_0$ and $\mU^{(k)}_1$ with equal prior probabilities in a single shot
is $(\gamma_k + 1) / 2$,
and thus $\gamma_k = 1$ holds if and only if $\mU^{(k)}_0$ and $\mU^{(k)}_1$ are
perfectly distinguishable.

Our main theorem is the following:
\begin{thm} \label{thm:main}
 For any integer $N \ge 1$, we have
 \begin{alignat}{1}
  P^{(N)} &= \frac{1}{N + 1} \left( \sum_{k=1}^N \gamma_k + 1 \right).
 \end{alignat}
\end{thm}

Theorem~\ref{thm:main_eq} in the main paper is a special case of this theorem
with $\gamma_1 = \cdots = \gamma_N = \gamma$.
In order to prove Theorem~\ref{thm:main},
we first show that there exists a tester such that
the success probability is $P^\star \coloneqq (\sum_{k=1}^N \gamma_k + 1) / (N + 1)$
in Sec.~\ref{subsec:thm_primal}.
This immediately implies $P^{(N)} \ge P^\star$.
In Sec.~\ref{subsec:thm_dual}, we next show that there exists a feasible solution, $X$,
to Problem~\eqref{prob:D} satisfying $\eta(X) = P^\star$,
which implies $P^{(N)} \le P^\star$ and thus $P^{(N)} = P^\star$.
In Sec.~\ref{subsec:thm_hamiltonian}, we prove Theorem~\ref{thm:main_hamiltonian} in the main paper
by applying Theorem~\ref{thm:main} to the problem of discriminating change points for Hamiltonians.

\subsection{Feasible solution to the primal problem} \label{subsec:thm_primal}

We assume that for each $k \in \{1,\dots,N\}$,
$\mU^{(k)}_0$ and $\mU^{(k)}_1$ are not perfectly distinguishable,
i.e., $\gamma_k < 1$ holds, unless otherwise stated.
Let $\mS$ be the set of all testers $\{ \mD_m \}_{m=0}^N$ such that
each $\mD_m \ast \mE_n$ is depicted by
\begin{alignat}{1}
 \InsertPDF{1.0}{process_discrimination_primal.pdf} ~\raisebox{1em}{,}
 \label{eq:process_discrimination_primal}
\end{alignat}
where
\begin{alignat}{1}
 \mQ^{(k,r)} &\coloneqq
 \begin{dcases}
  \Ad_{U^{(k,r)}_0{}^\dagger}, & r = R, \\
  \Ad_{Q^{(k,r)} U^{(k,r)}_0{}^\dagger}, & r < R \\
 \end{dcases}
\end{alignat}
and $Q^{(k,r)}$ is a unitary channel satisfying
$Q^{(k,r)} \ket{\lambda^{(k,r)}_b} = \ket{\lambda^{(k,r+1)}_b}$ $~(\forall b \in \{0,1\})$.
In such a tester $\{ \mD_m \}_{m=0}^N \in \mS$,
we first prepare the composite system $\V_{N,1} \ot \cdots \ot \V_{1,1}$
in an initial state $\rho$.
Just after applying the channel $\mU^{(k,r)}_\bin{n<k}$,
the channel $\mQ^{(k,r)} \in \Chn(\W_{k,r},\V_{k,r+1})$ is applied.
For each $r \in \{1,\dots,R-1\}$, the output of $\mQ^{(k,r)}$ is sent through
the channel $\mU^{(k,r+1)}_\bin{n<k}$.
We have
\begin{alignat}{1}
 \mQ^{(k,r)} \ast \mU^{(k,r)}_b &=
 \begin{dcases}
  \Ad_{U^{(k,R)}_0{}^\dagger U^{(k,R)}_b} = \Ad_{\tU^{(k,R)}_b}, & r = R, \\
  \Ad_{Q^{(k,r)} U^{(k,r)}_0{}^\dagger U^{(k,r)}_b} = \Ad_{Q^{(k,r)} \tU^{(k,r)}_b}, & r < R. \\
 \end{dcases}
 \label{eq:QU}
\end{alignat}
This can be explained as follows.
$\mQ^{(k,r)}$ transforms the channels $\mU^{(k,r)}_0$ and $\mU^{(k,r)}_1$ into
$\Ad_{\tU^{(k,r)}_0} = \ident_{\V_{k,r}}$ and $\Ad_{\tU^{(k,r)}_1}$, respectively;
furthermore, if $r < R$, then $\mQ^{(k,r)}$ transforms the normalized eigenvectors
$\ket{\lambda^{(k,r)}_0}$ and $\ket{\lambda^{(k,r)}_1}$ of $\tU^{(k,r)}_1$
into the normalized eigenvectors $\ket{\lambda^{(k,r+1)}_0}$ and $\ket{\lambda^{(k,r+1)}_1}$
of $\tU^{(k,r+1)}_1$, respectively.
Let
\begin{alignat}{1}
 \tU^{(k)}_b &\coloneqq \tU^{(k,R)}_b Q^{(k,R-1)} \tU^{(k,R-1)}_b Q^{(k,R-2)} \tU^{(k,R-2)}_b \cdots
 Q^{(k,1)} \tU^{(k,1)}_b;
\end{alignat}
then, it follows from Eq.~\eqref{eq:QU} that the channel
$\mQ^{(k,R)} \ast \mU^{(k,R)}_\bin{n<k} \ast \mQ^{(k,R-1)} \ast \mU^{(k,R-1)}_\bin{n<k} \ast
\cdots \ast \mQ^{(k,1)} \ast \mU^{(k,1)}_\bin{n<k}$ is equal to $\Ad_{\tU^{(k)}_\bin{n<k}}$.
We also have
\begin{alignat}{1}
 \tU^{(k)}_l \ket{\lambda^{(k,1)}_b} &=
 \begin{dcases}
  \ket{\lambda^{(k,R)}_b}, & l = 0, \\
  \lambda^{(k)}_b \ket{\lambda^{(k,R)}_b}, & l = 1. \\
 \end{dcases}
 \label{eq:tU_lambda}
\end{alignat}
A collective measurement, $\{ \Pi_m \}_{m=0}^N$, is performed on the outputs of $N$ channels
$\mQ^{(1,R)},\dots,\mQ^{(N,R)}$.
Assume that $\rho$ is a pure state $\ket{\psi}\bra{\psi}$.
Since each $\mQ^{(k,r)}$ is fixed, $\{ \mD_m \}_{m=0}^N$ is completely characterized by
$\ket{\psi}$ and $\{ \Pi_m \}_{m=0}^N$.
$\{ \mD_m \}_{m=0}^N$ is nonadaptive when $R = 1$.

Let
\begin{alignat}{1}
 \w_k &\coloneqq \sqrt{\lambda^{(k)}_0 \lambda^{(k)}_1}, \quad
 \zeta_k \coloneqq \frac{1 - \gamma_k}{1 + \gamma_k},
 \quad k \in \{1,\dots,N\}, \nonumber \\
 \nu_{k'} &\coloneqq \frac{1 - \sqrt{\zeta_{k'} \zeta_{k'+1}}}{1 + \sqrt{\zeta_{k'} \zeta_{k'+1}}},
 \quad k' \in \{1,\dots,N-1\}.
 \label{eq:wk}
\end{alignat}
Note that if
$\lambda^{(k)}_0$ and $\lambda^{(k)}_1$ are independent of $k$,
in which case $\gamma_1 = \cdots = \gamma_N \eqqcolon \gamma$ holds,
then $\omega_k$, $\zeta_k$, and $\nu_k$ are also independent of $k$,
and thus we have $\nu_k = \gamma$.
Let $\ket{\psi'_n}$ be the state just prior to the measurement $\{ \Pi_m \}_{m=0}^N$
when the change point is $n$; then, we can see
$\ket{\psi'_n} = \left[ \bigotimes_{k=0}^{N-1} \tU^{(N-k)}_\bin{n<N-k} \right] \ket{\psi}$.
Also, let $G \in \Complex^{(N+1) \times (N+1)}$ be the Gram matrix of
the collection of the states $\ket{\psi'_0}, \dots, \ket{\psi'_N}$,
whose $(n+1)$-th row and $(n'+1)$-th column is
$G_{n,n'} \coloneqq \braket{\psi'_n|\psi'_{n'}}$.

First, we prove the following two lemmas.
\begin{lemma} \label{lemma:Gram}
 Assume that $\gamma_k < 1$ holds for each $k \in \{1,\dots,N\}$; then,
 there exists a tester $\{ \mD_m \}_{m=0}^N \in \mS$ such that
 the Gram matrix $G$ of $\{ \ket{\psi'_0}, \dots, \ket{\psi'_N} \}$ satisfies
 \begin{alignat}{1}
  G_{n,n'} &= \left( \sum_{l=n'+1}^n \gamma_l + 1 \right) \prod_{l=n'+1}^n \sqrt{\zeta_l} \w_l,
  \quad \forall n \ge n'.
  \label{eq:Gram}
 \end{alignat}
\end{lemma}
\begin{proof}
 We will show that $G$ satisfies Eq.~\eqref{eq:Gram} when we choose
 \begin{alignat}{1}
  \ket{\psi} &\coloneqq \frac{1}{\sqrt{2^{N-1}}} \sum_{s_1 \in \{-1,1\}} \dots \sum_{s_N \in \{-1,1\}}
  \iota(s_1,\dots,s_N)
  \left[ \prod_{k=1}^{N-1} \frac{s_k s_{k+1} \sqrt{\nu_k} + 1}{\sqrt{\nu_k + 1}} \right]
  \frac{\ket{\lambda^{(N,1)}_0} + s_N \ket{\lambda^{(N,1)}_1}}{\sqrt{2}} \cdots
  \frac{\ket{\lambda^{(1,1)}_0} + s_1 \ket{\lambda^{(1,1)}_1}}{\sqrt{2}},
  \label{eq:psi0}
 \end{alignat}
 where $\iota(s_1,\dots,s_N)$ is $1$ if the number of $s_i$ equal to $-1$ is even,
 and $0$ otherwise.
 We can rewrite Eq.~\eqref{eq:psi0} as
 \begin{alignat}{1}
  \ket{\psi} &= \frac{1}{\sqrt{2}} \sum_{b_1=0}^1 \dots \sum_{b_N=0}^1
  \left( \prod_{k=1}^{N-1} \sqrt{\frac{\nu_k^{|b_{k+1}-b_k|}}{\nu_k + 1}} \right)
  \ket{\lambda^{(N,1)}_{b_N}} \cdots \ket{\lambda^{(1,1)}_{b_1}}.
  \label{eq:psi}
 \end{alignat}
 We have that, from Eq.~\eqref{eq:tU_lambda},
 \begin{alignat}{1}
  \ket{\psi'_n} &= \left[ \bigotimes_{k=0}^{N-1} \tU^{(N-k)}_\bin{n<N-k} \right] \ket{\psi}
  = \frac{1}{\sqrt{2}} \sum_{b_1=0}^1 \dots \sum_{b_N=0}^1
  \left[ \prod_{l=n+1}^N \lambda^{(l)}_{b_l} \right]
  \left( \prod_{k=1}^{N-1} \sqrt{\frac{\nu_k^{|b_{k+1}-b_k|}}{\nu_k + 1}} \right)
  \ket{\lambda^{(N,R)}_{b_N}} \cdots \ket{\lambda^{(1,R)}_{b_1}}.
  \label{eq:psi_}
 \end{alignat}
 It follows from Eq.~\eqref{eq:psi_} that $G$ is a Hermitian matrix satisfying
 $G_{n,n} = 1$ for each $n$ and
 \begin{alignat}{1}
  G_{n,n'} &= \frac{1}{2} \sum_{b_1=0}^1 \dots \sum_{b_N=0}^1
  \left[ \prod_{l=n'+1}^n \lambda^{(l)}_{b_l} \right]
  \prod_{k=1}^{N-1} \frac{\nu_k^{|b_{k+1}-b_k|}}{\nu_k + 1}
  = \frac{1}{2} \sum_{b_{n'+1}=0}^1 \dots \sum_{b_n=0}^1
  \left[ \prod_{l=n'+1}^n \lambda^{(l)}_{b_l} \right] 
  \prod_{k=n'+1}^{n-1} \frac{\nu_k^{|b_{k+1}-b_k|}}{\nu_k + 1}
  \label{eq:Gjk}
 \end{alignat}
 for each $n > n'$, where the second equality follows from
 \begin{alignat}{1}
  \sum_{b_k=0}^1 \frac{\nu_k^{|b_{k+1}-b_k|}}{\nu_k + 1} = 1, \qquad
  \sum_{b_{k+1}=0}^1 \frac{\nu_k^{|b_{k+1}-b_k|}}{\nu_k + 1} = 1.
 \end{alignat}
 Transforming the form of the sum of products of
 $\prod_{l=n'+1}^n \lambda^{(l)}_{b_l}$ into
 the form of the sum of products of
 $\prod_{l=n'+1}^n \left[ \lambda^{(l)}_1 \pm \lambda^{(l)}_0 \right]$
 yields
 \begin{alignat}{1}
  \sum_{b_{n'+1}=0}^1 \dots \sum_{b_n=0}^1
  \left[ \prod_{l=n'+1}^n \lambda^{(l)}_{b_l} \right]
  \prod_{k=n'+1}^{n-1} \nu_k^{|b_{k+1}-b_k|}
  &= \frac{1}{2^{n-n'-1}} \sum_{s_{n'} \in S_{n'}} \sum_{s_{n'+1} \in S_{n'+1}} \dots \sum_{s_n \in S_n}
  \left[ \prod_{l=n'+1}^n \left[ \lambda^{(l)}_1 + s_{l-1}s_l \lambda^{(l)}_0 \right] \right]
  \prod_{k=n'+1}^{n-1} (s_k \nu_k + 1)
  \label{eq:Gjk2}
 \end{alignat}
 with
 \begin{alignat}{1}
  S_k &\coloneqq
  \begin{dcases}
   \{-1,1\}, & n' < k < n, \\
   \{1\}, & \text{otherwise}. \\
  \end{dcases}
 \end{alignat}
 From Eq.~\eqref{eq:wk}, we have
 \begin{alignat}{1}
  \frac{s_k \nu_k + 1}{\nu_k + 1} &=
  \begin{dcases}
   1, & s_k = 1, \\
   \sqrt{\zeta_k \zeta_{k+1}}, & s_k = -1,
  \end{dcases} \nonumber \\
  \lambda^{(l)}_1 + \lambda^{(l)}_0 &= 2 \sqrt{\zeta_l} \w_l (1 + \gamma_l), \nonumber \\
  \lambda^{(l)}_1 - \lambda^{(l)}_0 &= 2 \i \w_l \gamma_l,
 \end{alignat}
 where $\i \coloneqq \sqrt{-1}$.
 Using these equations, we can rewrite Eq.~\eqref{eq:Gjk} as
 \begin{alignat}{1}
  G_{n,n'} &= 
  \frac{1}{2^{n-n'}} \sum_{s_{n'} \in S_{n'}} \sum_{s_{n'+1} \in S_{n'+1}} \dots \sum_{s_n \in S_n}
  \left[ \prod_{l=n'+1}^n \left[ \lambda^{(l)}_1 + s_{l-1}s_l \lambda^{(l)}_0 \right] \right]
  \prod_{k=n'+1}^{n-1} \frac{s_k \nu_k + 1}{\nu_k + 1} \nonumber \\
  &= g_{n-n'}(\gamma_{n'+1},\dots,\gamma_n) \prod_{l=n'+1}^n \sqrt{\zeta_l} \w_l,
  \label{eq:Gjk3}
 \end{alignat}
 where
 \begin{alignat}{1}
  g_k(x_1,\dots,x_k) &\coloneqq \sum_{s_1 \in \{-1,1\}} \dots \sum_{s_k \in \{-1,1\}}
  \beta_{1,s_1}(x_1) \prod_{l=2}^k \beta_{s_{l-1},s_l}(x_l), \nonumber \\
  \beta_{s,s'}(x) &\coloneqq \frac{ss' + 1}{2} + s' x =
  \begin{dcases}
   1 - x, & s = -1 ~\text{and}~ s' = -1, \\
   x, & s = -1 ~\text{and}~ s' = 1, \\
   - x, & s = 1 ~\text{and}~ s' = -1, \\
   1 + x, & s = 1 ~\text{and}~ s' = 1. \\
  \end{dcases}
 \end{alignat}
 Since
 \begin{alignat}{1}
  g_k(x_1,\dots,x_k) &= (x_k + 1) g_{k-1}(x_1,\dots,x_{k-1})
  - x_k g_{k-1}(x_1,\dots,x_{k-2},x_{k-1}-1), \quad k \ge 2
 \end{alignat}
 and $g_1(x_1) = x_1 + 1$ hold, we have
 \begin{alignat}{1}
  g_k(x_1,\dots,x_k) &= \sum_{l=1}^k x_l + 1.
 \end{alignat}
 Substituting this into Eq.~\eqref{eq:Gjk3} gives Eq.~\eqref{eq:Gram}.
\end{proof}

\begin{lemma} \label{lemma:pmn}
 Assume that given states $\{ \ket{\psi'_0}, \dots, \ket{\psi'_N} \}$ have
 the Gram matrix $G$ determined by Eq.~\eqref{eq:Gram}; then,
 there exists an orthonormal system $\{ \ket{\pi_0}, \dots, \ket{\pi_N} \}$
 satisfying
 \begin{alignat}{1}
  |\braket{\pi_m|\psi'_n}|^2 &= p_{m|n} \coloneqq \frac{\gamma_m + \gamma_{m+1}}{2}
  \prod_{l=\min(m,n)+1}^{\max(m,n)} \zeta_l.
  \label{eq:pmn}
 \end{alignat}
\end{lemma}
\begin{proof}
 Let $\ket{0},\dots,\ket{N}$ be an orthonormal system of
 $\V_{N,R+1} \ot \cdots \ot \V_{1,R+1}$.
 For each $n \in \{0,\dots,N\}$, let
 \begin{alignat}{1}
  \ket{\varphi_n} &\coloneqq \sum_{k=0}^N \sqrt{p_{k|n}} \left( \prod_{l=1}^k \w_l \right)
  \left( \prod_{l'=1}^n \w_{l'}^* \right) \ket{k},
 \end{alignat}
 where $^*$ is the complex conjugate.
 Assume that the Gram matrix of $\{ \ket{\varphi_0},\dots,\ket{\varphi_N}\}$ is
 equal to $G$ of Eq.~\eqref{eq:Gram};
 then, it follows that there exists a unitary matrix $\Upsilon$ satisfying
 $\ket{\psi'_n} = \Upsilon \ket{\varphi_n}$ for each $n$
 (see, e.g., Sec.~7.3 of Ref.~\cite{Hor-Joh-2012}).
 Fix such $\Upsilon$ and let
 \begin{alignat}{1}
  \ket{\pi_m} &\coloneqq \Upsilon \ket{m}.
  \label{eq:pim}
 \end{alignat}
 Note that $\{ \ket{\pi_0}, \dots, \ket{\pi_N} \}$ is obviously an orthonormal system.
 Since
 \begin{alignat}{1}
  \ket{\psi'_n} &= \Upsilon \ket{\varphi_n} = \sum_{k=0}^N \sqrt{p_{k|n}} \left( \prod_{l=1}^k \w_l \right)
  \left( \prod_{l'=1}^n \w_{l'}^* \right) \ket{\pi_m}
 \end{alignat}
 holds, we have Eq.~\eqref{eq:pmn}.

 It remains to show that the Gram matrix of $\{ \ket{\varphi_0},\dots,\ket{\varphi_N}\}$
 is equal to $G$.
 Since the Gram matrix is Hermitian, it suffices to show
 $\braket{\varphi_n|\varphi_{n'}} = G_{n,n'}$ $~(\forall n \ge n')$.
 We have that, for each $n$ and $n'$ with $n \ge n'$,
 \begin{alignat}{1}
  \braket{\varphi_n|\varphi_{n'}} &= \sum_{k=0}^N \sqrt{p_{k|n} p_{k|n'}} \prod_{t=n'+1}^n \w_t
  = \left( \prod_{t=n'+1}^n \sqrt{\zeta_t} \w_t \right) \sum_{k=0}^N \frac{\gamma_k + \gamma_{k+1}}{2}
  \left( \prod_{l=k+1}^{n'} \zeta_l \right) \left( \prod_{l'=n+1}^k \zeta_{l'} \right).
  \label{varphi_nn}
 \end{alignat}
 We have that, from $(1 + \gamma_k) \zeta_k = 1 - \gamma_k$,
 \begin{alignat}{1}
  \sum_{k=0}^{n'} \frac{\gamma_k + \gamma_{k+1}}{2}
  \left( \prod_{l=k+1}^{n'} \zeta_l \right) \left( \prod_{l'=n+1}^k \zeta_{l'} \right)
  &= \sum_{k=0}^{n'} \frac{\gamma_k + \gamma_{k+1}}{2}
  \left( \prod_{l=k+1}^{n'} \zeta_l \right) \nonumber \\
  &= \frac{1 - \gamma_1}{2} \left( \prod_{l=2}^{n'} \zeta_l \right)
  + \sum_{k=1}^{n'} \frac{\gamma_k + \gamma_{k+1}}{2}
  \left( \prod_{l=k+1}^{n'} \zeta_l \right) \nonumber \\
  &= \frac{1 - \gamma_2}{2} \left( \prod_{l=3}^{n'} \zeta_l \right)
  + \sum_{k=2}^{n'} \frac{\gamma_k + \gamma_{k+1}}{2}
  \left( \prod_{l=k+1}^{n'} \zeta_l \right)
  = \dots = \frac{1 + \gamma_{n'+1}}{2}, \nonumber \\
  \sum_{k=n}^N \frac{\gamma_k + \gamma_{k+1}}{2}
  \left( \prod_{l=k+1}^{n'} \zeta_l \right) \left( \prod_{l'=n+1}^k \zeta_{l'} \right)
  &= \sum_{k=n}^N \frac{\gamma_k + \gamma_{k+1}}{2}
  \left( \prod_{l'=n+1}^k \zeta_{l'} \right) \nonumber \\
  &= \frac{1 - \gamma_N}{2} \left( \prod_{l'=n+1}^{N-1} \zeta_{l'} \right)
  \sum_{k=n}^{N-1} \frac{\gamma_k + \gamma_{k+1}}{2}
  \left( \prod_{l'=n+1}^k \zeta_{l'} \right) \nonumber \\
  &= \frac{1 - \gamma_{N-1}}{2} \left( \prod_{l'=n+1}^{N-2} \zeta_{l'} \right)
  \sum_{k=n}^{N-2} \frac{\gamma_k + \gamma_{k+1}}{2}
  \left( \prod_{l'=n+1}^k \zeta_{l'} \right)
  = \dots = \frac{1 + \gamma_n}{2}.
 \end{alignat}
 Thus,
 \begin{alignat}{1}
  \sum_{k=0}^N \frac{\gamma_k + \gamma_{k+1}}{2}
  \left( \prod_{l=k+1}^{n'} \zeta_l \right) \left( \prod_{l'=n+1}^k \zeta_{l'} \right)
  &= \frac{1 + \gamma_{n'+1}}{2} + \sum_{k=n'+1}^{n-1} \frac{\gamma_k + \gamma_{k+1}}{2}
  + \frac{1 + \gamma_n}{2}
  = \sum_{k=n'+1}^n \gamma_l + 1
 \end{alignat}
 holds.
 Substituting this into Eq.~\eqref{varphi_nn} and using Eq.~\eqref{eq:Gram}
 gives $\braket{\varphi_n|\varphi_{n'}} = G_{n,n'}$.
\end{proof}

Next, using the above lemmas, we prove the following lemma.
\begin{lemma} \label{lemma:primal}
 There exists a tester $\{ \mD_m \}_{m=0}^N \in \mS$
 whose success probability is $(\sum_{k=1}^N \gamma_k + 1) / (N + 1)$.
\end{lemma}
\begin{proof}
 First, we consider the case in which, for any $k \in \{1,\dots,N\}$,
 $\mU^{(k)}_0$ and $\mU^{(k)}_1$ are not perfectly distinguishable,
 i.e., $\gamma_k < 1$ holds.
 Let us consider the tester $\{ \mD_m \}_{m=0}^N \in \mS$
 that is determined by the input state $\ket{\psi}$ given by Eq.~\eqref{eq:psi0}
 and the measurement $\{ \Pi_m \}_{m=0}^N$
 with $\Pi_m \coloneqq \ket{\pi_m} \bra{\pi_m}$ $~(\forall m \in \{0,\dots,N-1\})$
 and $\Pi_N \coloneqq \ident_{\V_{N,R+1} \ot \cdots \ot \V_{1,R+1}} - \sum_{m=0}^{N-1} \Pi_m$,
 where $\ket{\pi_m}$ is given by Eq.~\eqref{eq:pim}.
 [Note that we can choose any $\{ \Pi_m \}_{m=0}^N$ such that $\Pi_m \ge \ket{\pi_m} \bra{\pi_m}$
 $~(\forall m \in \{0,\dots,N\})$ and
 $\sum_{m=0}^N \Pi_m = \ident_{\V_{N,R+1} \ot \cdots \ot \V_{1,R+1}}$.]
 From Lemmas~\ref{lemma:Gram} and \ref{lemma:pmn},
 the conditional probability that the measurement result is $m$ given that
 the change point is $n$ is given by $|\braket{\pi_m|\psi'_n}|^2 = p_{m|n}$.
 Thus, the average probability is
 \begin{alignat}{1}
  \frac{1}{N + 1} \sum_{n=0}^N p_{n|n} &=
  \frac{1}{N + 1} \sum_{n=0}^N \frac{\gamma_n + \gamma_{n+1}}{2}
  = \frac{1}{N + 1} \left( \sum_{n=1}^N \gamma_n + 1 \right).
 \end{alignat}
 
 Next, we consider the case in which there exists $\tk$ satisfying $\gamma_{\tk} = 1$.
 In the following discussion, we assume for simplicity that $\gamma_k < 1$ holds
 for each $k \neq \tk$, but the same discussion can be easily applied to any other case.
 Since $\mU^{(\tk)}_0$ and $\mU^{(\tk)}_1$ are perfectly distinguishable,
 we can distinguish $n < \tk$ and $n \ge \tk$ with certainty.
 Thus, Problem~\eqref{prob:P} can be divided into the problem of
 discriminating $\tk$ processes
 $\{ \mU^{(\tk-1)}_\bin{n < \tk-1} \ast \cdots \ast \mU^{(1)}_\bin{n < 1} \}_{n=0}^{\tk-1}$
 and the problem of discriminating $N + 1 - \tk$ processes
 $\{ \mU^{(N)}_\bin{n < N} \ast \cdots \ast \mU^{(\tk+1)}_\bin{n < \tk + 1} \}_{n=\tk}^N$.
 As mentioned above, it follows that, for the former and latter problems,
 there exist testers with the success probabilities
 $(\sum_{n=1}^{\tk-1} \gamma_n + 1) / \tk$ and
 $(\sum_{n=\tk+1}^N \gamma_n + 1) / (N + 1 - \tk)$, respectively.
 Thus, for Problem~\eqref{prob:P}, there exists a tester whose success probability is
 \begin{alignat}{1}
  \frac{\tk}{N + 1} \cdot \frac{1}{\tk} \left( \sum_{n=1}^{\tk-1} \gamma_n + 1 \right)
  + \frac{N + 1 - \tk}{N + 1} \cdot \frac{1}{N + 1 - \tk} \left( \sum_{n=\tk+1}^N \gamma_n + 1 \right)
  &= \frac{1}{N + 1} \left( \sum_{n=1}^N \gamma_n + 1 \right).
 \end{alignat}
\end{proof}

\subsection{Feasible solution to the dual problem} \label{subsec:thm_dual}

For each $n \in \{1,\dots,N\}$, let $Y^{(n)}$ be an optimal solution
to the following optimization problem:
\begin{alignat}{1}
 \begin{array}{ll}
  \mbox{minimize} & \eta(Y) \\
  \mbox{subject~to} & \displaystyle Y \in \tmC^{(n)},
   ~Y \ge \frac{\ChoiU^{(n)}_0}{2}, ~Y \ge \frac{\ChoiU^{(n)}_1}{2}, \\
 \end{array}
 \label{prob:DY}
\end{alignat}
where
\begin{alignat}{1}
 \tmC^{(n)} &\coloneqq \Cone(\Comb_{\W_{n,R},\V_{n,R};\dots;\W_{n,1},\V_{n,1}}).
\end{alignat}
This problem can be seen as the Lagrange dual of the problem of
discriminating the two processes $\ChoiU^{(n)}_0$ and $\ChoiU^{(n)}_1$ with equal prior probabilities.
As already mentioned at the beginning of this section,
the optimal value of this primal problem is $(\gamma_n + 1) / 2$.
Since the duality gap is zero, we have $\eta[Y^{(n)}] = (\gamma_n + 1) / 2$.

Let
\begin{alignat}{1}
 X^{(n)} &\coloneqq
 \begin{dcases}
  Y^{(1)}, & n = 1, \\
  \frac{1}{n+1} \left[ \ChoiU^{(n)}_1 \ot n X^{(n-1)}
  + \left[ 2 Y^{(n)} - \ChoiU^{(n)}_1 \right] \ot \ChoiE^{(n-1)}_{n-1} \right],
  & n \ge 2; \\
 \end{dcases}
 \label{eq:Xn}
\end{alignat}
then, we have the following lemma.
\begin{lemma} \label{lemma:XN}
 $X^{(N)}$ is a feasible solution to Problem~\eqref{prob:D}.
\end{lemma}
\begin{proof}
 We prove by induction on $n$ that
 $X^{(n)}$ is a feasible solution to Problem~\eqref{prob:D} with $N = n$.
 The case $n = 1$ is obvious from $X^{(1)} = Y^{(1)}$.
 Suppose now that, for given $n \ge 2$, $X^{(n-1)}$ is a feasible solution to
 Problem~\eqref{prob:D} with $N = n - 1$;
 we have to show $X^{(n)} \in \mC^{(n)}$ and
 $X^{(n)} \ge \ChoiE^{(n)}_m / (n+1)$ $~(\forall m \in \{0,\dots,n\})$.

 First, we show $X^{(n)} \in \mC^{(n)}$.
 We can verify that $Z \ot Z'$ is in $\mC^{(n)}$ for any
 $Z \in \tmC^{(n)}$ and $Z' \in \mC^{(n-1)}$.
 Thus, since $\ChoiU^{(n)}_1, 2Y^{(n)} - \ChoiU^{(n)}_1 \in \tmC^{(n)}$ and
 $n X^{(n-1)}, \ChoiE^{(n-1)}_{n-1} \in \mC^{(n-1)}$ hold,
 we have $X^{(n)} \in \mC^{(n)}$.

 Next, we show $X^{(n)} \ge \ChoiE^{(n)}_m / (n+1)$ $~(\forall m \in \{0,\dots,n\})$.
 In the case of $m < n$, we have that, from Eq.~\eqref{eq:Xn},
 \begin{alignat}{1}
  X^{(n)} &\ge \frac{\ChoiU^{(n)}_1 \ot n X^{(n-1)}}{n+1}
  \ge \frac{\ChoiU^{(n)}_1 \ot \ChoiE^{(n-1)}_m}{n+1} = \frac{\ChoiE^{(n)}_m}{n+1},
 \end{alignat}
 where the first inequality follows from $Y^{(n)} \ge \ChoiU^{(n)}_1 / 2$,
 and the equality follows from $\ChoiU^{(n)}_1 \ot \ChoiE^{(n-1)}_m = \ChoiE^{(n)}_m$,
 which is obtained from Eq.~\eqref{eq:tmEN}.
 In the case of $m = n$, we have that, from Eq.~\eqref{eq:Xn},
 \begin{alignat}{1}
  X^{(n)} &\ge \frac{2 Y^{(n)} \ot \ChoiE^{(n-1)}_{n-1}}{n+1}
  \ge \frac{\ChoiU_0 \ot \ChoiE^{(n-1)}_{n-1}}{n+1} = \frac{\ChoiE^{(n)}_n}{n+1},
 \end{alignat}
 where the first and second inequalities follow from
 $X^{(n-1)} \ge \ChoiE^{(n-1)}_{n-1} / n$ and
 $Y^{(n)} \ge \ChoiU^{(n)}_0 / 2$, respectively.
 The equality follows from $\ChoiU^{(n)}_0 \ot \ChoiE^{(n-1)}_{n-1} = \ChoiE^{(n)}_n$,
 which is obtained from Eq.~\eqref{eq:tmEN}.
\end{proof}

We now prove Theorem~\ref{thm:main} by using Lemmas~\ref{lemma:primal} and \ref{lemma:XN}.
Let $q_n \coloneqq \eta[X^{(n)}]$ and $y_n \coloneqq \eta[Y^{(n)}]$.
We have that, from \eqref{eq:Xn},
\begin{alignat}{1}
 q_n &=
 \begin{dcases}
  y_1, & n = 1, \\
  \frac{n q_{n-1} + 2y_n - 1}{n+1}, & n \ge 2, \\
 \end{dcases}
\end{alignat}
which follows from $\eta(Z \ot Z') = \eta(Z)\eta(Z')$
$~[\forall Z \in \tmC^{(n)}, ~Z' \in \mC^{(n-1)}]$,
$\eta[\ChoiU^{(n)}_1] = 1$, and $\eta[\ChoiE^{(n-1)}_{n-1}] = 1$.
Thus, we have
\begin{alignat}{1}
 q_N &= \frac{1}{N+1} \left( \sum_{k=1}^N 2 y_k - N + 1 \right).
 \label{eq:qN}
\end{alignat}
Substituting $y_n = (\gamma_n + 1) / 2$ into this equation yields
$q_N = (\sum_{k=1}^N \gamma_k + 1) / (N + 1)$.
Since $X^{(N)}$ is a feasible solution to Problem~\eqref{prob:D},
$P^{(N)} \le \eta[X^{(N)}] = q_N$ holds.
On the other hand, from Lemma~\ref{lemma:primal},
there exists a tester whose success probability is $q_N$,
and thus $P^{(N)} \ge q_N$ holds.
Therefore, we have $P^{(N)} = q_N$.

\subsection{Proof of Theorem~\ref{thm:main_hamiltonian}} \label{subsec:thm_hamiltonian}

Let us consider the problem of discriminating change points of Hamiltonians
$H_0(t)$ and $H_1(t)$, which is addressed in the main paper.
Assume, without loss of generality, that the candidates, $t_0,\dots,t_N$, of the change point
can be expressed in the form $t_k = \tau_{kR}$,
where $R$ is a natural number and $\tau_0,\tau_1,\dots,\tau_{NR}$ are real numbers
satisfying $\tau_0 < \tau_1 < \dots < \tau_{NR}$.
Let $\mU^{(k,r)}_b$ be the unitary channel defined by
\begin{alignat}{1}
 \mU^{(k,r)}_b(\rho) &\coloneqq U^{(k,r)}_b \rho U^{(k,r)}_b{}^\dagger,
 \quad U^{(k,r)}_b \coloneqq \mT \left[ \exp \left[ - \i \int_{\tau_{(k-1)R + r - 1}}^{\tau_{(k-1)R + r}}
 H_b(t) dt \right] \right],
 \quad b \in \{0,1\}, ~k \in \{1,\dots,N\}, ~r \in \{1,\dots,R\},
\end{alignat}
where $\mT$ is the time order operator.
Also, $\mU^{(k)}_b$ be defined by Eq.~\eqref{eq:Ukb}.
This problem can be regarded as the problem of discriminating quantum combs
$\mE^{(N)}_0,\dots,\mE^{(N)}_N$ expressed by Eq.~\eqref{eq:E}.

From Lemma~\ref{lemma:XN} and Theorem~\ref{thm:main},
$X^{(N)}$ of Eq.~\eqref{eq:Xn} is an optimal solution to Problem~\eqref{prob:D}.
$X^{(N)}$ is obtained from the optimal solutions $Y^{(1)},\dots,Y^{(N)}$ of Problem~\eqref{prob:DY}.
From Eq.~\eqref{eq:qN}, for each $k$, the greater $y_k = \eta[Y^{(k)}]$
[i.e., the success probability of distinguishing $\ChoiU^{(k)}_0$ and $\ChoiU^{(k)}_1$],
the greater $q_N = \eta[X^{(N)}]$.
$y_k$ approaches a maximal value of $(\gamma_k + 1) / 2$ in the limit $R \to \infty$,
where
\begin{alignat}{1}
 \gamma_k &\coloneqq \Delta \left[
 \int_{t_{k-1}}^{t_k} [\mu_\mathrm{max}(t) - \mu_\mathrm{min}(t)] dt \right],
\end{alignat}
$\Delta(\theta) \coloneqq \sin[\min(\theta,\pi)/2]$,
and  $\mu_\mathrm{max}(t)$ and $\mu_\mathrm{min}(t)$ are, respectively,
the maximum and minimum eigenvalues of $H_1(t) - H_0(t)$ \cite{Chi-Pre-Ren-1999,Aha-Mas-Pop-2002}.
In this limit, we have that, from Eq.~\eqref{eq:qN},
\begin{alignat}{1}
 q_N &= \frac{1}{N + 1} \left( \sum_{k=1}^N \gamma_k + 1 \right).
\end{alignat}
Moreover, it follows from Lemma~\eqref{lemma:primal} in the limit $R \to \infty$ that
there exists a tester whose success probability is $q_N$.
Thus, the maximum success probability is $(\sum_{k=1}^N \gamma_k + 1) / (N + 1)$.

The conditional probability that the outcome of such an optimal tester is $t_m$ given that
the change point is $t_n$ is
\begin{alignat}{1}
 \pH_{m|n} &\coloneqq \frac{\gamma_m + \gamma_{m+1}}{2} \prod_{l=\min(m,n)+1}^{\max(m,n)} \zeta_l,
 \label{eq:pH2}
\end{alignat}
which is of the same form as Eq.~\eqref{eq:pmn}.
Note that if, as in Ref.~\cite{Aha-Mas-Pop-2002}, the Hilbert space $\mH_{Nbox}$
corresponding to ``the particle passing next to the box'' can be used,
then we can consider a process $\mU^{(k)}_b \oplus \ident_{\mH_{Nbox}}$,
instead of $\mU^{(k)}_b$.

An optimal tester is expressed as Eq.~\eqref{eq:process_discrimination_primal} in the limit
$R \to \infty$.
Such a tester consists of the following three steps:
\begin{enumerate}[label=\arabic*)]
 \item Prepare the composite system $\V_{N,1} \ot \cdots \ot \V_{1,1}$
       in the pure state $\ket{\psi}$ of Eq.~\eqref{eq:psi0}.
 \item Repeat the following procedure for each $k = 1,\dots,N$:
       during the time interval from $t_{k-1}$ to $t_k$,
       let the Hamiltonian $H_\bin{n < k}(t) - H_0(t)$ act on the system $\V_{k,1}$,
       where $H_\bin{n < k}(t)$ is the given Hamiltonian and $-H_0(t)$ is a control Hamiltonian.
       Simultaneously, apply a unitary transformation that maps the eigenstates
       of $H_1(t - 0) - H_0(t - 0)$ to the corresponding eigenstates of $H_1(t + 0) - H_0(t + 0)$.
 \item Perform a measurement $\{ \Pi_m \}_{m=0}^N$ on the output state at time $t_N$,
       where $\Pi_m \coloneqq \ket{\pi_m} \bra{\pi_m}$ $~(\forall m \in \{0,\dots,N-1\})$,
       $\Pi_N \coloneqq \ident - \sum_{m=0}^{N-1} \Pi_m$,
       and $\{ \ket{\pi_m} \}_{m=0}^N$ is an orthonormal system satisfying Eq.~\eqref{eq:pmn}.
\end{enumerate}

\section{Relationship between inner products of states and maximum success probability}
\label{sec:iprod}

As a simple example, we consider the problem of discriminating three pure states,
denoted by $\ket{\psi_0}$, $\ket{\psi_1}$, and $\ket{\psi_2}$, with equal prior probabilities.
Assume $\braket{\psi_1|\psi_0} = \braket{\psi_2|\psi_1} = 0.8$ and
$\braket{\psi_2|\psi_0} = a$ with some nonnegative real number $a$; then,
the Gram matrix of these states is
\begin{alignat}{1}
 \begin{bmatrix}
  1 & 0.8 & a \\
  0.8 & 1 & 0.8 \\
  a & 0.8 & 1 \\
 \end{bmatrix}.
\end{alignat}
Since the Gram matrix is positive semidefinite, $0.28 \le a \le 1$ must hold.
The maximum success probability, denoted by $P$, is obtained from the Gram matrix.
The smaller $a$ is, the greater is the maximum success probability for discriminating 
$\ket{\psi_0}$ and $\ket{\psi_2}$.
So one might expect that $P$ reaches its maximum value for the smallest $a$.
However, this is not true as we can see from Fig.~\ref{fig:result_Gram3}.
\begin{figure}[h]
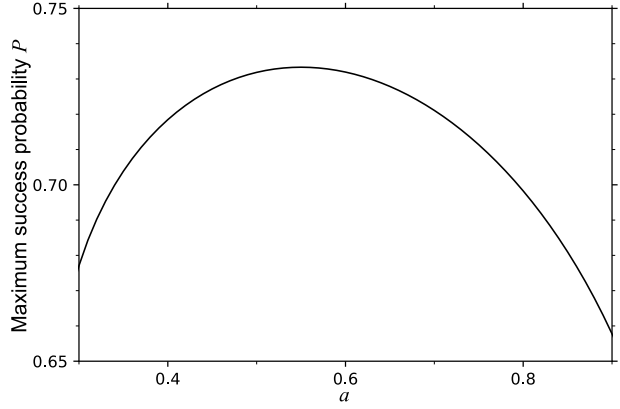

 \centering
 \InsertPDF{0.6}{result-Gram3.pdf}
 \caption{Maximum success probability, $P$, versus $a$.}
 \label{fig:result_Gram3}
\end{figure}

\section{Classical postprocessing} \label{sec:range}

In the problem of discriminating change points of Hamiltonians,
the conditional probability, $\pH_{m|n}$, of guessing the change point $t_m$ given that
it is actually $t_n$ is given by Eq.~\eqref{eq:pH2}.
Suppose here that after the discrimination, we got the information that
the change point $t_n$ is not smaller than $t_{n_0}$ and not greater than $t_{n_1}$,
(i.e., $t_{n_0} \le t_n \le t_{n_1}$), where $0 \le n_0 < n_1 \le N$ holds.
Also, suppose that the new prior probabilities of the possible change points
$t_{n_0},\dots,t_{n_1}$ are all equal.
Let $t_m$ be the discrimination result.
In such a case, we can still achieve optimal discrimination when
%
\begin{alignat}{1}
 t'_m &\coloneqq
 \begin{dcases}
  t_m, & n_0 \le m \le n_1, \\
  t_{n_0}, & m < n_0, \\
  t_{n_1}, & m > n_1
 \end{dcases}
\end{alignat}
is simply used as the final estimate.
Indeed, from $(1 + \gamma_k) \zeta_k = 1 - \gamma_k$,
the conditional probability, denoted by $p'_{m|n}$, of obtaining the final estimate $t'_m$ given that
it is actually $t_n$ satisfies
\begin{alignat}{1}
 p'_{n_0|n} &= \sum_{k=0}^{n_0} \pH_{k|n}
 = \frac{1 + \gamma_1}{2} \prod_{l=1}^n \zeta_l
 + \sum_{k=1}^{n_0} \frac{\gamma_k + \gamma_{k+1}}{2} \prod_{l=k+1}^n \zeta_l
 = \frac{1 + \gamma_2}{2} \prod_{l=2}^n \zeta_l
 + \sum_{k=2}^{n_0} \frac{\gamma_k + \gamma_{k+1}}{2} \prod_{l=k+1}^n \zeta_l
 = \dots = \frac{1 + \gamma_{n_0+1}}{2} \prod_{l=n_0+1}^n \zeta_l, \nonumber \\
 p'_{n_1|n} &= \sum_{k=n_1}^N \pH_{k|n}
 = \frac{\gamma_N + 1}{2} \prod_{l=n+1}^N \zeta_l
 + \sum_{k=n_1}^{N-1} \frac{\gamma_k + \gamma_{k+1}}{2} \prod_{l=n+1}^k \zeta_l
 = \frac{\gamma_{N-1} + 1}{2} \prod_{l=n+1}^{N-1} \zeta_l
 + \sum_{k=n_1}^{N-2} \frac{\gamma_k + \gamma_{k+1}}{2} \prod_{l=n+1}^k \zeta_l
 = \dots = \frac{\gamma_{n_1} + 1}{2} \prod_{l=n+1}^{n_1} \zeta_l.
\end{alignat}
This indicates that $p'_{m|n}$ is equal to $\pH_{m|n}$ with $\gamma_{n_0} = \gamma_{n_1 + 1} = 1$.
The success probability is
\begin{alignat}{1}
 \frac{1}{n_1 - n_0 + 1} \sum_{m=n_0}^{n_1} p'_{m|n}
 &= \frac{1}{n_1 - n_0 + 1} \left( \sum_{k=n_0 + 1}^{n_1} \gamma_k + 1 \right),
\end{alignat}
and thus this strategy is optimal.

\section{Maximum likelihood estimator for the change point} \label{sec:MLE}

In the problem discussed in Sec.~\ref{sec:problem},
we here consider the case $1 > \gamma_1 = \gamma_2 = \dots = \gamma_N \eqqcolon \gamma$.
Assume for simplicity that the change point runs over all integers.
Let $\zeta \coloneqq (1-\gamma) / (1+\gamma)$.
For the optimal discrimination mentioned in Sec.~\ref{sec:thm},
it follows from Eq.~\eqref{eq:pmn} that
the conditional probability, $p_{m|n}$, that the measurement result is $m$ given that
the change point is $n$ is
\begin{alignat}{1}
 p_{m|n} &= \mL(m - n; \zeta) = \frac{1-\zeta}{1+\zeta} \zeta^{|m-n|}.
\end{alignat}
Thus, for each $n$, $\{ p_{m|n} = \mL(m-n; \zeta) \}_m$ has the discrete Laplace distribution.
Let us repeat the same experiment $L \ge 2$ times using the same optimal discrimination strategy
and assume that $L$ discrimination results, denoted by $m_1,\dots,m_L$, are obtained.
We want to determine the change point $n$ as accurately as possible from $m_1,\dots,m_L$.
This can be seen as the classical problem of finding the maximum likelihood estimate of $n$.
The probability of obtaining the results $m_1,\dots,m_L$ is given as
\begin{alignat}{1}
 \mathrm{Prob}(m_1,\dots,m_L;n) &\coloneqq \prod_{l=1}^L \mL(m_l - n; \zeta)
 = \left( \frac{1-\zeta}{1+\zeta} \right)^L \zeta^{\sum_{l=1}^L |m_l-n|}.
\end{alignat}
The maximum likelihood estimate of $n$ is given by
\begin{alignat}{1}
 \hat{n} &\in \argmax_{n} \mathrm{Prob}(m_1,\dots,m_L;n)
 = \argmin_{n} \sum_{l=1}^L |m_l-n|,
\end{alignat}
where the equality follows from $0 < \zeta < 1$.
Note that we assume here that the prior probabilities of the change points are all equal,
in which case the maximum likelihood estimate is equivalent to the maximum a posteriori estimate.

Let $m'_1,\dots,m'_L$ be the discrimination results sorted in ascending order.
It follows that for any natural number $l$ not greater than $L/2$,
$|m'_l - m| + |m'_{L+1-l} - m|$ attains its minimum value of $m'_{L+1-l} - m'_l$
when $m'_l \le m \le m'_{L+1-l}$ holds.
Thus, the maximum likelihood estimate $\hat{n}$ is $m'_{(L+1)/2}$
(i.e., the median of $m_1,\dots,m_L$) if $L$ is odd and
any integer $m$ satisfying $m'_{L/2} \le m \le m'_{L/2+1}$ if $L$ is even.

\Letter{}{
 \bibliographystyle{apsrev4-1}
 \input{bibliography.bbl}
}

\end{document}

%% file: settings_quant.tex
\usepackage{bm}

\theorembodyfont{\upshape}
\newtheorem{thmMain}{Theorem}
\newtheorem{thm}{Theorem}

\newtheorem{postulateno}{Postulate}

\newtheorem{lemma}[thm]{Lemma}

\newcounter{proof}

\ExplSyntaxOn
\NewDocumentEnvironment{proof}{o}
 {
  \par\medskip
  \noindent
  \textbf{Proof~}
 }
 {\QED\par\smallskip}
\ExplSyntaxOff

\newcounter{postulate}
\renewcommand{\thepostulate}{\arabic{postulate}}
\ExplSyntaxOn
\NewDocumentEnvironment{postulate}{oo}
 {
  \refstepcounter{postulate}
  \begin{postulateno}
  \textbf{\hspace{-0.5em}\IfNoValueTF{#2}{\thepostulate}{#2} ~(\IfNoValueTF{#1}{}{#1})}
 }
 {
  \end{postulateno}
 }
\ExplSyntaxOff

\newcommand{\mC}{\mathcal{C}}
\newcommand{\mD}{\mathcal{D}}
\newcommand{\mE}{\mathcal{E}}

\newcommand{\mH}{\mathcal{H}}

\newcommand{\mL}{\mathcal{L}}

\newcommand{\mQ}{\mathcal{Q}}

\newcommand{\mS}{\mathcal{S}}
\newcommand{\mT}{\mathcal{T}}
\newcommand{\mU}{\mathcal{U}}

\newcommand{\mX}{\mathcal{X}}

\newcommand{\ident}{\hat{1}}

\newcommand{\Real}{\mathbb{R}}
\newcommand{\Complex}{\mathbb{C}}

\newcommand{\QED}{\hspace*{0pt}\hfill $\blacksquare$}

\DeclareMathOperator*{\argmax}{argmax}
\DeclareMathOperator*{\argmin}{argmin}
\renewcommand{\i}{\mathbf{i}}
\newcommand{\T}{\mathsf{T}}
\DeclareMathOperator{\Tr}{Tr}
\newcommand{\Trp}[1]{\mathop{\mathrm{Tr}_{#1}}}

\DeclareMathAlphabet{\mymathbb}{U}{BOONDOX-ds}{m}{n}

\newcommand{\kket}[1]{|#1\rangle\!\rangle}
\newcommand{\bbra}[1]{\langle\!\langle#1|}

%% file: bibliography.bbl
%